\documentclass[11pt]{article}
\usepackage{lmodern}
\usepackage[T1]{fontenc}
\usepackage{color}
\usepackage{float}
\usepackage{calc}
\usepackage{amsmath,amsthm,amsfonts,amssymb,bm}
\usepackage{mathtools}
\usepackage{url}
\usepackage[ruled,vlined,linesnumbered,algo2e]{algorithm2e}
\usepackage{nicefrac}
\usepackage{xifthen}

\global\long\def\R{\mathbb{R}}%
\global\long\def\Ex{\mathbb{E}}%
\global\long\def\cf{\mathcal{F}}%
\global\long\def\cs{\mathcal{S}}%
\global\long\def\opt{\mathrm{OPT}}%
\global\long\def\stg{\textup{\textsc{STGreedy}}}%
\global\long\def\offlinepp{\textsc{PostProcess}}%
\global\long\def\offlinealg{\textup{\textsc{OfflineAlg}}}%
\global\long\def\stream{\textup{\textsc{StreamProcessRandomized}}\xspace}%
\global\long\def\streamcont{\textup{\textsc{StreamProcessExtension}}\xspace}%
\global\long\def\streamdeterministic{\textup{\textsc{StreamProcess}}\xspace}%
\global\long\def\x{\mathbf{x}}%
\global\long\def\one{\mathbf{1}}%
\global\long\def\threshold{\kappa} 

\DeclareMathOperator{\supp}{supp}

\newcommand{\RFunc}{\mathtt{R}}
\newcommand{\groundset}{V}
\newcommand{\realnum}{\mathbb{R}}
\newcommand{\nonnegative}{\realnum_{\geq 0}}

\newcommand{\onevect}[1]{\ensuremath{\mathbf{1}_{#1}}}
\newcommand{\onenorm}[1]{\ensuremath{\left\|{#1}\right\|_1}}
\newcommand{\range}[2]{\ensuremath{\bracks{#1,#2}}}
\newcommand{\cube}[3]{\ensuremath{\range{#1}{#2}^{#3}}}
\newcommand{\Exp}[1]{\ensuremath{\mathbb{E}\bracks{#1}}}
\newcommand{\pr}[1]{\ensuremath{\Pr\bracks{#1}}}
\newcommand{\Func}[3]{\ensuremath{#1\colon#2\rightarrow#3}}

\newcommand{\DeclareAutoPairedDelimiter}[3]{%
  \expandafter\DeclarePairedDelimiter\csname Auto\string#1\endcsname{#2}{#3}%
  \begingroup\edef\x{\endgroup
    \noexpand\DeclareRobustCommand{\noexpand#1}{%
      \expandafter\noexpand\csname Auto\string#1\endcsname*}}%
  \x}

\DeclareAutoPairedDelimiter\power{\lvert}{\rvert}
\DeclareAutoPairedDelimiter{\ceil}{\lceil}{\rceil}
\DeclareAutoPairedDelimiter{\bracks}{[}{]}
\DeclareAutoPairedDelimiter{\curly}{\{}{\}}

\makeatletter
\newtheorem*{rep@theorem}{\rep@title}
\newcommand{\newreptheorem}[2]{%
\newenvironment{rep#1}[1]{%
 \def\rep@title{#2 \ref{##1}}%
 \begin{rep@theorem}}%
 {\end{rep@theorem}}}
\makeatother

\newreptheorem{theorem}{Theorem}
\newreptheorem{corollary}{Corollary}
\newreptheorem{observation}{Observation}
\newreptheorem{lemma}{Lemma}
\newreptheorem{proposition}{Proposition}

\newtheorem{theorem}{Theorem}
\newtheorem{lemma}[theorem]{Lemma}
\newtheorem{corollary}[theorem]{Corollary}
\newtheorem{observation}[theorem]{Observation}
\newtheorem{proposition}[theorem]{Proposition}

\begin{document}
\title{An Optimal Streaming Algorithm for Submodular Maximization with a Cardinality Constraint\thanks{An earlier version of this paper appeared in ICALP 2020~\cite{DBLP:conf/icalp/AlalufEFNS20}.}}
\author{
Naor Alaluf\thanks{Department of Mathematics and Computer Science, Open University of Israel, \texttt{naoralaluf@gmail.com}.}
\and
Alina Ene\thanks{Department of Computer Science, Boston University, \texttt{aene@bu.edu}.}
\and
Moran Feldman\thanks{Department of Computer Science, University of Haifa, \texttt{moranfe@cs.haifa.ac.il}.}
\and
Huy L. Nguyen\thanks{Khoury College of Computer and Information Science, Northeastern University, \texttt{hlnguyen@cs.princeton.edu}.}
\and
Andrew Suh\thanks{Department of Computer Science, Boston University, \texttt{asuh9@bu.edu}.}
}
\maketitle

\begin{abstract}
We study the problem of maximizing a \emph{non-monotone} submodular
function subject to a cardinality constraint in the streaming model. Our
main contribution is a single-pass \mbox{(semi-)}\allowbreak streaming algorithm that uses roughly $O(k / \varepsilon^2)$
memory, where $k$ is the size constraint. At the end of the stream,
our algorithm post-processes its data structure using \emph{any offline
algorithm} for submodular maximization, and obtains a solution whose
approximation guarantee is $\frac{\alpha}{1+\alpha}-\varepsilon$, where
$\alpha$ is the approximation of the offline algorithm. If we use
an exact (exponential time) post-processing algorithm, this leads to
$\frac{1}{2}-\varepsilon$ approximation (which is nearly optimal). If we post-process with the algorithm of \cite{buchbinder2019constrained},
that achieves the state-of-the-art offline approximation guarantee of $\alpha=0.385$,
we obtain $0.2779$-approximation in polynomial time, improving
over the previously best polynomial-time approximation of $0.1715$ due to \cite{feldman2018less}. It is also worth mentioning that our algorithm is combinatorial and deterministic, which is rare for an algorithm for non-monotone submodular maximization, and enjoys a fast update time of $O(\frac{\log k + \log (1/\alpha)}{\varepsilon^2})$ per element.
\end{abstract}

\newpage{}
\section{Introduction}

In this paper, we study the problem of maximizing a \emph{non-monotone}
submodular function subject to a cardinality (size) constraint in the streaming
model. This problem captures problems of interest in a wide-range
of domains, such as machine learning, data mining, combinatorial optimization,
algorithmic game theory, social networks, and many others. A representative
application is data summarization, where the goal is to select a
small subset of the data that captures the salient features of the
overall dataset \cite{badanidiyuru2014}. One can model these problems
as submodular maximization with a cardinality constraint: the submodular
objective captures how informative the summary is, as well as other
considerations such as how diverse the summary is, and the cardinality constraint
ensures that the summary is small. Obtaining such a summary is very
beneficial when working with massive data sets, that may not even fit
into memory, since it makes it possible to analyze the data using algorithms
that would be prohibitive to run on the entire dataset.

There have been two main approaches to deal with the large size of
modern data sets: the \emph{distributed} computation approach partitions
the data across many machines and uses local computation on the machines
and communication across the machines in order to perform the analysis,
and the \emph{streaming} computation approach processes the data in
a stream using only a small amount of memory and (ideally) only a single
pass over the data. Classical algorithms for submodular maximization,
such as the Greedy algorithm, are not suitable in these settings since
they are centralized and require many passes over the data. Motivated
by the applications as well as theoretical considerations, there has
been a significant interest in studying submodular maximization problems
both in the distributed and the streaming setting, leading to many
new results and insights \cite{Kumar2013,Mirzasoleiman2013,badanidiyuru2014,chakrabarti2015submodular,chekuri2015streaming,MirrokniZ15,Barbosa2015,Mirzasoleiman2015distributed,Barbosa2016,epasto2017bicriteria,mirzasoleiman2018streaming,feldman2018less,norouzi2018beyond}.

Despite this significant progress, several fundamental questions remain
open both in the streaming and distributed setting. In the streaming
setting, which is the main focus of this paper, submodular maximization
is fairly well understood when the objective function is additionally
\emph{monotone}---i.e., we have $f(A)\leq f(B)$ whenever $A\subseteq B$.
For example, the Greedy approach, which obtains an optimal $(1-1/e)$-approximation in the centralized setting when the function is monotone~\cite{Nemhauser1978},
can be adapted to the streaming model~\cite{Kumar2013,badanidiyuru2014}.
This yields the single-threshold Greedy algorithm: make a single pass over
the data and select an item if its marginal gain exceeds a suitably
chosen threshold. If the threshold is chosen to be $\frac{1}{2}\frac{f(\opt)}{k}$,
where $f(\opt)$ is the value of the optimal solution and $k$ is
the cardinality constraint, then the single-threshold Greedy algorithm
is guaranteed to achieve $\frac{1}{2}$-approximation. Although
the value of the optimal solution is unknown, it can be estimated
based on the largest singleton value even in the streaming setting~\cite{badanidiyuru2014}. The algorithm uses the maximum singleton value to make $O(\varepsilon^{-1} \log k)$ guesses for $f(\opt)$ and, for each guess, it runs the single-threshold Greedy algorithm; which leads to $(\frac{1}{2}-\varepsilon)$-approximation.
Remarkably, this approximation guarantee is optimal in
the streaming model even if we allow unbounded computational power:
Feldman \emph{et al. }\cite{feldman2020oneway} showed that any
algorithm for this problem that achieves an approximation
better than $\frac{1}{2}$ requires $\Omega\left(\frac{n}{k^3}\right)$
memory, where $n$ is the length of the stream. Additionally, the single-threshold Greedy algorithm
enjoys a fast update time of $O(\varepsilon^{-1} \log k)$
marginal value computations per item and uses only $O(\varepsilon^{-1} k\log k)$
space.\footnote{A variant of the algorithm due to~\cite{DBLP:conf/icml/0001MZLK19} has an even better space complexity of $O(k/\varepsilon)$.}

In contrast, the general problem with a \emph{non-monotone} objective
has proved to be considerably more challenging. Even in the centralized
setting, the Greedy algorithm fails to achieve any constant approximation guarantee
when the objective is non-monotone. Thus, several
approaches have been developed for handling non-monotone objectives in this setting, including local search \cite{feige2011maximizing,lee2010submodular,lee2009non}, continuous optimization \cite{Feldman2011a,ene2016constrained,buchbinder2019constrained}
and sampling \cite{DBLP:conf/soda/BuchbinderFNS14,DBLP:conf/colt/FeldmanHK17}. 
The currently best approximation guarantee is $0.385$ \cite{buchbinder2019constrained},
the strongest inapproximability is $0.491$ \cite{Gharan2011}, and it remains
a long-standing open problem to settle the approximability of submodular
maximization subject to a cardinality constraint.

Adapting the above techniques to the streaming setting is challenging, and the approximation guarantees are weaker. The main approach for non-monotone maximization in the streaming setting has been to extend the local search algorithm of Chakrabarti and Kale \cite{chakrabarti2015submodular} from monotone to non-monotone objectives. This approach was employed in a sequence of works \cite{chekuri2015streaming,feldman2018less,mirzasoleiman2018streaming}, leading to the currently best approximation of $\frac{1}{3+2\sqrt{2}}\approx0.1715$.\footnote{Chekuri et al.~\cite{chekuri2015streaming} claimed an improved approximation ratio of $\frac{1}{2+e}-\varepsilon$ for a cardinality constraint, but an error was later found in the proof of this improved ratio~\cite{communication:Chekuri18}. See Appendix~\ref{app:error} for more details.}
This naturally leads to the following questions.
\begin{itemize}
\item \emph{What is the optimal approximation ratio achievable for submodular
maximization in the streaming model?} \emph{In particular, is it possible to achieve $\frac{1}{2}-\varepsilon$ approximation using an algorithm that uses only $\mathrm{poly}(k,1/\varepsilon)$
space?}
\item \emph{Is there a good streaming algorithm for non-monotone functions based on the single-threshold
Greedy algorithm that works so well for monotone functions?}
\item \emph{Can we exploit existing heuristics for the offline problem in
the streaming setting?}
\end{itemize}

\textbf{Our contributions.}
In this work, we give an affirmative answer
to all of the above questions. Specifically, we give a streaming algorithm\footnote{Formally, all the algorithms we present are semi-streaming algorithms, i.e., their space complexity is nearly linear in $k$. Since this is unavoidable for algorithms designed to output an approximate solution (as opposed to just estimating the value of the optimal solution), we ignore the difference between streaming and semi-streaming algorithms in this paper and use the two terms interchangeably.} that performs a single pass over the stream and outputs sets of size $O(k / \varepsilon)$ that can be post-processed
using \emph{any offline algorithm} for submodular maximization. The
post-processing is itself quite straightforward: we simply run the
offline algorithm on the output set to obtain a solution of size at
most $k$. We show that, if the offline algorithm achieves $\alpha$-approximation, then we obtain $\left(\frac{\alpha}{1+\alpha}-\varepsilon\right)$-approximation.
One can note that if we post-process using an exact (exponential
time) algorithm, we obtain $(\frac{1}{2}-\varepsilon)$-approximation. This matches the inapproximability result proven by~\cite{feldman2020oneway} for the special case of a monotone objective function. Furthermore, we show that in the non-monotone case any streaming algorithm guaranteeing $(\frac{1}{2}+\varepsilon)$-approximation for some positive constant $\varepsilon$ must use in fact $\Omega(n)$ space.\footnote{This result is a simple adaptation of a result due to Buchbinder et al.~\cite{BuchbinderFS19}. For completeness, we include the proof in Appendix~\ref{app:inapproximability}.} Thus, we essentially settle the approximability of the problem if exponential-time computation is allowed.

The best (polynomial-time) approximation guarantee that is currently
known in the offline setting is $\alpha=0.385$ \cite{buchbinder2019constrained}.
If we post-process using this algorithm, we obtain $0.2779$-approximation
in polynomial time, improving over the previously best polynomial-time
approximation of $0.1715$ due to \cite{feldman2018less}. The offline
algorithm of \cite{buchbinder2019constrained} is based on the multilinear
extension, and thus is quite slow. One can obtain
a more efficient overall algorithm by using the combinatorial random Greedy
algorithm of \cite{DBLP:conf/soda/BuchbinderFNS14} that achieves $\frac{1}{e}$-approximation. Furthermore, any existing heuristic for the offline
problem can be used for post-processing, exploiting their effectiveness
beyond the worst case.

\textbf{Variants of our algorithm.}
Essentially, every algorithm for non-monotone submodular maximization includes a randomized component. Oftentimes this component is explicit, but in some cases it takes more subtle forms such as maintaining multiple solutions that are viewed as the support of a distribution~\cite{DBLP:journals/talg/BuchbinderF18,DBLP:conf/colt/FeldmanHK17} or using the multilinear extensions (which is defined via expectations)~\cite{Feldman2011a,buchbinder2019constrained}. We present in this work three variants of our algorithm based on the above three methods of introducing a randomized component into the algorithm.

Perhaps the most straightforward way to introduce a randomized component into the single-threshold Greedy algorithm is to use the multilinear extension as the objective function and include only fractions of elements in the solution (which corresponds to including the elements in the solution only with some bounded probability). This has the advantage of keeping the algorithm almost deterministic (in fact, completely deterministic when the multilinear extension can be evaluated deterministically), which allows for a relatively simple analysis of the algorithm and a low space complexity of $O(k\log \alpha^{-1} / \varepsilon^2)$. However, the time complexity of an algorithm obtained via this approach depends on the complexity of evaluating the multilinear extension, which in general can be quite high. In Appendix~\ref{app:continuous}, we describe and analyze a variant of our algorithm (named {\streamcont}) which is based on the multilinear extension.

To avoid the multilinear extension, and its associated time complexity penalty, one can use true randomization, and pass every arriving element to single-threshold Greedy only with a given probability. However, analyzing such a combination of single-threshold Greedy with true randomization is difficult because it requires delicate care of the event that the single-threshold Greedy algorithm fills up the budget. In particular, this was the source of the subtle error mentioned above in one of the results of \cite{chekuri2015streaming}. Our approach for handling this issue is to consider two cases depending on the probability that the budget is filled up in a run (this is a good event since the resulting solution has good value). If this probability is sufficiently large (at least $\varepsilon$), we repeat the basic algorithm $O(\ln(1/\varepsilon)/\varepsilon)$ times in parallel to boost the probability of this good event to $1-\varepsilon$. Otherwise, the probability that the budget is not filled up in a run is at least $1-\varepsilon$, and conditioning on this event changes the probabilities by only a $1-\varepsilon$ factor. Another issue with true randomness is that some elements can be randomly discarded despite being highly desirable. Following ideas from distributed algorithms~\cite{MirrokniZ15,Barbosa2015,Barbosa2016}, this issue can also be solved by running multiple copies of the algorithm in parallel since such a run guarantees that with high probability every desirable element is processed by some copy. Using this approach we get a variant of our algorithm named {\stream} which can be found in Appendix~\ref{app:randomized}. {\stream} is combinatorial and fast (it has an update time of $\tilde{O}(\varepsilon^{-2})$ marginal value computations per element), but due to the heavy use of parallel runs, it has a slightly worse space complexity of $\tilde{O}(k/\varepsilon^{3})$.

The two above discussed variants of our algorithm appeared already in an earlier conference version of this paper~\cite{DBLP:conf/icalp/AlalufEFNS20}. Our main result, however, is a new variant of our algorithm (named \streamdeterministic) based on the technique of maintaining multiple solutions and treating them as the support of a distribution. In retrospect, creating a variant based on this technique is natural since it combines the advantages of the two previous approaches. Specifically, we get an algorithm which is deterministic and combinatorial, has a relatively simple analysis and enjoys a low space complexity of $O(k \varepsilon^{-2} \log \alpha^{-1})$ and a low update time of $O(\varepsilon^{-2} \log (k/\alpha))$ marginal value computations per element.

While the last variant of our algorithm has the best time and space guarantees, we give also the two earlier variants for two reasons. The first reason is that they demonstrate the first use of techniques such as continuous extensions and random partitions in the context of streaming algorithms for submodular maximization, and thus, greatly expand the toolkit available in this context. These techniques have proved to be quite versatile in the sequential and distributed settings, and we hope that they will lead to further developments in the streaming setting as well. The second reason is that, while the asymptotic approximation guarantees of all three variants of our algorithm are identical given a black box offline $\alpha$-approximation algorithm, they might differ when the offline algorithm has additional properties. For example, the approximation ratio of the offline algorithm might depend on the ratio between $k$ and the number of elements in its input (see~\cite{DBLP:conf/soda/BuchbinderFNS14} for an example of such an algorithm). Given such an offline algorithm, it might be beneficial to set the parameters controlling the number of elements passed to the offline algorithm so that only a moderate number of elements is passed, which is a regime in which the three variants of our algorithm produce different approximation guarantees. Hence, one of the first two variants of our algorithm might end up having a better approximation guarantee than that of the last variant if future research yields offline algorithms with relevant properties.

\subsection{Additional related work}

The problem of maximizing a non-negative monotone submodular function subject to a cardinality or a matroid constraint was studied (in the offline model) already in the $1970$'s. In $1978$, Nemhauser et al.~\cite{DBLP:journals/mp/NemhauserWF78} and Fisher et al.~\cite{Fisher1978} showed that a natural greedy algorithm achieves an approximation ratio of $1 - \nicefrac{1}{e} \approx 0.632$ for this problem when the constraint is a cardinality constraint and an approximation ratio of $\nicefrac{1}{2}$ for matroid constraints. The $1 - \nicefrac{1}{e}$ approximation ratio for cardinality constraints was shown to be optimal already on the same year by Nemhauser and Wolsey~\cite{DBLP:journals/mor/NemhauserW78}, but the best possible approximation ratio for matroid constraints was open for a long time. Only a decade ago, Calinescu et al.~\cite{DBLP:journals/siamcomp/CalinescuCPV11} managed to show that a more involved algorithm, known as ``continuous greedy'', can achieve $(1 - \nicefrac{1}{e})$-approximation for this type of constraint, which is tight since matroid constraints generalize cardinality constraints.

Unlike the natural greedy algorithm, continuous greedy is a randomized algorithm, which raised an interesting question regarding the best possible approximation ratio for matroid constraints that can be achieved by a deterministic algorithm. Very recently, Buchbinder et al.~\cite{DBLP:conf/soda/BuchbinderF019} made a slight step towards answering this question. Specifically, they described a deterministic algorithm for maximizing a monotone submodular function subject to a matroid constraint whose approximation ratio is $0.5008$. This algorithm shows that the $\nicefrac{1}{2}$-approximation of the greedy algorithm is not the right answer for the above mentioned question.

Many works have studied also the offline problem of maximizing a non-negative (not necessarily monotone) submodular function subject to a cardinality or matroid constraint~\cite{buchbinder2019constrained,DBLP:conf/soda/BuchbinderFNS14,DBLP:journals/siamcomp/ChekuriVZ14,ene2016constrained,DBLP:journals/talg/Feldman17,Feldman2011a}. The most recent of these works achieves an approximation ratio of $0.385$ for both cardinality and matroid constraints~\cite{buchbinder2019constrained}. In contrast, it is known that no polynomial time algorithm can achieve an approximation ratio of $0.497$ for cardinality constraints or $0.478$ for matroid constraints, respectively~\cite{Gharan2011}.

The study of streaming algorithms for submodular maximization problems is related to the study of online algorithms for such problems. A partial list of works on algorithms of the last kind includes~\cite{DBLP:conf/esa/AzarGR11,BuchbinderFG19,BuchbinderFS19,DBLP:journals/talg/ChanHJKT18,DBLP:conf/soda/KapralovPV13,DBLP:journals/siamcomp/KorulaMZ18}.

\section{Preliminaries} \label{sec:preliminaries}

\textbf{Basic notation.} Let $\groundset$ denote a finite ground set of
size $n:=|\groundset|$. We occasionally assume without loss of generality that
$\groundset=\{1,2,\dots,n\}$, and use, e.g., $x=(x_{1},x_{2},\dots,x_{n})$
to denote a vector in $\R^{\groundset}$. For two vectors $x, y\in\R^{\groundset}$,
we let $x \vee y$ and $x \wedge y$ be the vectors such that $(x \vee y)_{e}=\max\{x_{e},y_{e}\}$ and $(x \wedge y)_{e}=\min\{x_{e},y_{e}\}$
for all $e \in \groundset$. For a set $S \subseteq \groundset$, we let $\one_{S}$ denote
the indicator vector of $S$, i.e., the vector that has $1$ in every
coordinate $e\in S$ and $0$ in every coordinate $e \in \groundset\setminus S$. Given an element $e \in \groundset$, we use $\one_{e}$ as a shorthand for $\one_{\{e\}}$. Furthermore, if $S$ is a random subset of $\groundset$, we use $\Exp{\one_{S}}$
to denote the vector $p$ such that $p_{e}=\pr{e\in S}$
for all $e \in \groundset$ (i.e., the expectation is applied coordinate-wise).

\textbf{Submodular functions.} In this paper, we consider the problem
of maximizing a non-negative submodular function subject to a cardinality
constraint. A set function $f\colon 2^{\groundset}\to\R$ is submodular if $f(A)+f(B)\geq f(A\cap B)+f(A\cup B)$
for all subsets $A,B\subseteq \groundset$. Additionally, given a set $S \subseteq \groundset$ and an element $e \in \groundset$, we use $f(e \mid S)$ to denote the marginal contribution of $e$ to $S$ with respect to the set function $f$, i.e., $f(e \mid S) = f(S \cup \{e\}) - f(S)$.

\textbf{Continuous extensions.} We make use of two standard continuous extensions of submodular functions. The first of these extensions is known as the \textit{multilinear extension}. To define this extension, we first need to define the random set $\RFunc(x)$. For every vector $x \in \cube{0}{1}{\groundset}$, $\RFunc(x)$ is defined as a random subset of $\groundset$ that includes every element $e \in \groundset$ with probability $x_e$, independently. The multilinear extension $F$ of $f$ is now defined for every
$x\in\cube{0}{1}{\groundset}$ by
\begin{align*}
	F(x)=\Exp{f\big(\RFunc(x)\big)}
	=\sum_{A\subseteq\groundset}{f(A)\cdot\pr{\RFunc(x)=A}}
	=\sum_{A\subseteq\groundset}{\left(f(A)\cdot\prod_{e\in{A}}{x_e}\cdot\prod_{e\notin{A}}{(1-x_e)}\right)}\enspace.
\end{align*}
One can observe from the definition that $F$ is indeed a multilinear function of the coordinates of $x$, as suggested by its name. Thus, if we use the shorthand $\partial_eF(x)$ for the first partial derivative
$\frac{\partial{F(x)}}{\partial{x_e}}$ of the multilinear extension $F$, then
$
	\partial_eF(x) = F(x\vee\onevect{e})-F\big(x\wedge\onevect{\groundset\setminus\curly{e}}\big)
$.

The second extension we need is known as the \textit{Lov\'{a}sz extension}. Unlike the multilinear extension that explicitly appears in one of the variants of our algorithm, the Lov\'{a}sz extension is not part of any of these variants. However, it plays a central role in the analyses of two of them. The Lov\'{a}sz extension $\hat{f}\colon [0,1]^{\groundset}\to\R$ is defined as follows. For every $x\in[0,1]^{\groundset}$,
$
\hat{f}(x)=\Ex_{\theta\sim[0,1]}[f(\{e\in \groundset \colon x_{e}\geq\theta\})] 
$,
where we use the notation $\theta\sim[0,1]$ to denote a value chosen
uniformly at random from the interval $[0,1]$. The Lov\'{a}sz extension
$\hat{f}$ of a non-negative submodular function has the following properties: (1) convexity: $c\hat{f}(x)+(1-c)\hat{f}(y)\geq \hat{f}(cx+(1-c)y)$
for all $x,y\in[0,1]^{\groundset}$ and all $c\in[0,1]$~\cite{DBLP:conf/ismp/Lovasz82}; (2) restricted
scale invariance: $\hat{f}(cx)\geq c\hat{f}(x)$ for all $x\in[0,1]^{\groundset}$
and all $c\in[0,1]$; (3) it lower bounds the multilinear extension, i.e., $F(x)\geq\hat{f}(x)$ for every $x\in\cube{0}{1}{\groundset}$ \cite[Lemma~A.4]{DBLP:journals/siamcomp/Vondrak13}.

\section{Simplified Algorithm}
\label{sec:simplified}

The properties of our main algorithm (\streamdeterministic---the third variant discussed above) are summarized by the following theorem.
\newcommand{\TrmMain}[1][]{Assume there exists an $\alpha$-approximation offline algorithm $\offlinealg$ for maximizing a non-negative submodular function subject to a cardinality constraint whose space complexity is nearly linear in the size of the ground set. Then, for every constant $\varepsilon \in (0,1]$, there exists an $(\frac{\alpha}{1 + \alpha} - \varepsilon)$-approximation semi-streaming algorithm for maximizing a non-negative submodular function subject to a cardinality constraint. The algorithm stores $O(k\varepsilon^{-2}\ifthenelse{\isempty{#1}}{\log \alpha^{-1}}{})$ elements and makes $O(\varepsilon^{-2}\log \ifthenelse{\isempty{#1}}{(k / \alpha)}{k})$ marginal value computations while processing each arriving element.\ifthenelse{\isempty{#1}}{}{\footnote{Formally, the number of elements stored by the algorithm and the number of marginal value computations also depend on $\log \alpha^{-1}$. Since $\alpha$ is typically a positive constant, or at least lower bounded by a positive constant, we omit this dependence from the statement of the theorem.}} Furthermore, if $\offlinealg$ is deterministic, then so is the algorithm that we get.} 
\begin{theorem} \label{trm:main}
\TrmMain[*]
\end{theorem}

In this section, we introduce a simplified version of the algorithm used to prove Theorem~\ref{trm:main}. This simplified version (given as Algorithm~\ref{alg:main_algorithm_deterministic}) captures our main new ideas, but makes the simplifying assumption that it has access to an estimate $\tau$ of $f(\opt)$ obeying $(1 - \varepsilon/2) \cdot f(\opt) \leq \tau \leq f(\opt)$. Such an estimate can be produced using well-known techniques, at the cost of a slight increase in the space complexity and update time of the algorithm. More specifically, in Section~\ref{sec:full} we formally show that one such technique due to~\cite{DBLP:conf/icml/0001MZLK19} can be used for that purpose, and that it increases the space complexity and update of the algorithm only by factors of $O(\varepsilon^{-1} \log \alpha^{-1})$ and $O(\varepsilon^{-1} \log (k/\alpha))$, respectively.

Algorithm~\ref{alg:main_algorithm_deterministic} gets two parameters: the approximation ratio $\alpha$ of $\offlinealg$ and an integer $p \geq 1$. The algorithm maintains $p$ solutions $S_1, S_2, \dotsc, S_p$. All these solutions start empty, and the algorithm may add each arriving element to at most one of them. Specifically, when an element $e$ arrives, the algorithm checks whether there exists a solution $S_i$ such that (1) $S_i$ does not already contain $k$ elements, and (2) the marginal contribution of $e$ with respect to $S_i$ exceeds the threshold of $c\tau / k$ for $c = \alpha / (1 + \alpha)$. If there is such a solution, the algorithm adds $e$ to it (if there are multiple such solutions, the algorithm adds $e$ to an arbitrary one of them); otherwise, the algorithm discards $e$. After viewing all of the elements, Algorithm~\ref{alg:main_algorithm} generates one more solution $S_o$ by executing $\offlinealg$ on the union of the $p$ solutions $S_1, S_2, \dotsc, S_p$. The output of the algorithm is then the best solution among the $p + 1$ generated solutions.

\begin{algorithm2e}[ht]
\caption{\streamdeterministic (simplified) $(p,\alpha)$} \label{alg:main_algorithm_deterministic}
	\DontPrintSemicolon
	Let $c\gets\frac{\alpha}{1 + \alpha}$.\\
	\lFor{$i = 1$ \KwTo $p$}{Let $S_i \gets \varnothing$.}
	\For {each arriving element $e$} {
		\If{there exists an integer $1 \leq i \leq p$ such that $|S_i| < k$ and $f(e \mid S_i) \geq \frac{c\tau}{k}$}
		{
			Update $S_i \gets S_i \cup \{e\}$ (if there are multiple options for $i$, pick an arbitrary one).
		}
	}
	Find another feasible solution $S_o \subseteq \bigcup_{i = 1}^p S_i$ by running $\offlinealg$ with $\bigcup_{i = 1}^p S_i$ as the ground set.\\
	\Return the solution maximizing $f$ among $S_o$ and the $p$ solutions $S_1, S_2, \dotsc, S_p$.
\end{algorithm2e}

Since Algorithm~\ref{alg:main_algorithm_deterministic} stores elements only in the $p + 1$ solutions it maintains, and all these solutions are feasible (and thus, contain only $k$ elements), we immediately get the following observation. Note that this observation implies (in particular) that Algorithm~\ref{alg:main_algorithm_deterministic} is a semi-streaming algorithm for a constant $p$ when the space complexity of $\offlinealg$ is nearly linear.
\begin{observation} \label{obs:space_simpleMain}
Algorithm~\ref{alg:main_algorithm_deterministic} stores at most $O(pk)$ elements, and makes at most $p$ marginal value calculations while processing each arriving element.
\end{observation}

We now divert our attention to analyzing the approximation ratio of Algorithm~\ref{alg:main_algorithm_deterministic}. Let us denote by $\hat{S}_i$ the final set $S_i$ (i.e., the content of this set when the stream ends), and consider two cases. The first (easy) case is when at least one of the solutions $S_1, S_2, \dotsc, S_p$ reaches a size of $k$. The next lemma analyzes the approximation guarantee of Algorithm~\ref{alg:main_algorithm_deterministic} in this case.
\begin{lemma}
\label{lem:Fx-exactlyKMain}
	If there is an integer $1 \leq i \leq p$ such that $|\hat{S}_i| = k$, then the output of Algorithm~\ref{alg:main_algorithm_deterministic} has value of at least $\frac{\alpha\tau}{1 + \alpha}$.
\end{lemma}
\begin{proof}
	Denote by $e_1,e_2,\dots,e_k$ the elements of $\hat{S}_i$ in the order of their arrival. Using this notation, the value of $f(\hat{S}_i)$ can be written as follows.
	\[
		f(\hat{S}_i)
		=
		f(\varnothing)+\sum_{j=1}^{k}{f\big(e_j \mid \curly{e_1,e_2,\dotsc,e_{j-1}}\big)}
		\geq
		\sum_{i=1}^{k} \frac{\alpha\tau}{(1 + \alpha)k}
		=
		\frac{\alpha\tau}{1 + \alpha}
		\enspace,
	\]
where the inequality holds since the non-negativity of $f$ implies $f(\varnothing) \geq 0$ and Algorithm~\ref{alg:main_algorithm_deterministic} adds an element $e_j$ to $S_i$ only when $f\big(e_j \mid \curly{e_1,e_2,\dotsc,e_{j-1}}\big)\geq\frac{c\tau}{k} = \frac{\alpha\tau}{(1 + \alpha)k}$. The lemma now follows since the solution outputted by Algorithm~\ref{alg:main_algorithm_deterministic} is at least as good as $S_i$.
\end{proof}

Consider now the case in which no set $S_i$ reaches the size of $k$. In this case our objective is to show that at least one of the solutions computed by Algorithm~\ref{alg:main_algorithm_deterministic} has a large value. Lemmata~\ref{lem:average} and~\ref{lem:belowKMain} lower bound the value of the average solution among $S_1, S_2, \dotsc, S_p$ and the solution $S_o$, respectively. The proof of Lemma~\ref{lem:average} uses the following known lemma.

\begin{lemma}[Lemma 2.2 from~\cite{DBLP:conf/soda/BuchbinderFNS14}]
\label{lem:atMostP}
	Let $\Func{f}{2^{\groundset}}{\nonnegative}$ be a non-negative submodular function. Denote by $A(p)$ a random subset of $A$
	where each element appears with probability at most $p$ (not necessarily independently). Then,
	$\Exp{f(A(p))}\geq(1-p) \cdot f(\varnothing)$.
\end{lemma}

Let $O = \opt \setminus \bigcup_{i = 1}^p \hat{S}_i$, and let $b = |O|/k$.
\begin{lemma} \label{lem:average}
If $|\hat{S}_i| < k$ for every integer $1 \leq i \leq p$, then, for every fixed set $A \subseteq \groundset$, $p^{-1} \cdot \sum_{i = 1}^p f(\hat{S}_i \cup A) \geq
	(1 - p^{-1}) \cdot f(O \cup A) - \alpha b \tau / (1 + \alpha)$.
\end{lemma}
\begin{proof}
The elements in $O$ were rejected by Algorithm~\ref{alg:main_algorithm_deterministic}. Since no set $S_i$ reaches a size of $k$, this means that the the marginal contribution of the elements of $O$ with respect to every set $S_i$ at the time of their arrival was smaller than $c\tau/k$. Moreover, since Algorithm~\ref{alg:main_algorithm_deterministic} only adds elements to its solutions during its execution, the submodularity of $f$ guarantees that the marginals of the elements of $O$ are below this threshold also with respect to $\hat{S}_1 \cup A, \hat{S}_2 \cup A, \dotsc, \hat{S}_p \cup A$. More formally, we get
\[
	f(e \mid \hat{S}_i \cup A) <\frac{c\tau}{k} = \frac{\alpha\tau}{k(1 + \alpha)}
	\quad \forall\; e\in  O \text{ and integer $1 \leq i \leq p$}
	\enspace.
\]
Using the submodularity of $f$ again, this implies that for every integer $1 \leq i \leq p$
	\[
		f\big(\hat{S}_i \cup O \cup A\big)
			\leq f(\hat{S}_i \cup A)+\sum_{e\in O} f(e \mid \hat{S}_i \cup A)
			\leq f(\hat{S}_i \cup A)+|O| \cdot \frac{\alpha\tau}{k(1 + \alpha)}
			= f(\hat{S}_i \cup A)+\frac{\alpha b\tau}{1 + \alpha}\enspace.
	\]
	
Adding up the above inequalities (and dividing by $p$), we get
\[
	p^{-1} \cdot \sum_{i = 1}^p f(\hat{S}_i \cup A)
	\geq
	p^{-1} \cdot \sum_{i = 1}^p f\big(\hat{S}_i \cup O \cup A\big) - \frac{\alpha b\tau}{1 + \alpha}
	\enspace.
\]
We now note that $p^{-1} \cdot \sum_{i = 1}^p f\big(\hat{S}_i \cup O \cup A\big)$ can be viewed as the expected value of a non-negative submodular function $g(S) = f(S \cup O \cup A)$ over a random set $S$ that is equal to every one of the sets $\hat{S}_1, \hat{S}_2, \dotsc, \hat{S}_p$ with probability $1/p$. Since the sets $\hat{S}_1, \hat{S}_2, \dotsc, \hat{S}_p$ are disjoint, $S$ contains every element with probability at most $1/p$, and thus, by Lemma~\ref{lem:atMostP},
\[
	p^{-1} \cdot \sum_{i = 1}^p f\big(\hat{S}_i \cup O \cup A\big)
	=
	\Exp{g(S)}
	\geq
	(1 - p^{-1}) \cdot g(\varnothing)
	=
	(1 - p^{-1}) \cdot f(O \cup A)
	\enspace.
\]
The lemma now follows by combining the last two inequalities.
\end{proof}

As mentioned above, our next step is to get a lower bound on the value of $f(S_o)$. One easy way to get such a lower bound is to observe that $\opt \setminus O$ is a subset of $\bigcup_{i = 1}^p \hat{S}_i$ of size at most $k$, and thus, is a feasible solution for the instance faced by the algorithm $\offlinealg$ used to find $S_o$; which implies $\Exp{f(S_o)} \geq \alpha \cdot f(\opt \setminus O)$ since  $\offlinealg$ is an $\alpha$-approximation algorithm. The following lemma proves a more involved lower bound by considering the vectors $(b \onevect{\hat{S}_i}) \vee \onevect{\opt \setminus O}$ as feasible fractional solutions for the same instance (using rounding methods such as Pipage Rounding or Swap Rounding~\cite{DBLP:journals/siamcomp/CalinescuCPV11,DBLP:conf/focs/ChekuriVZ10}, such feasible factional solutions can be converted into integral feasible solutions of at least the same value).

\begin{lemma} \label{lem:belowKMain}
If $|\hat{S}_i| < k$ for every integer $1 \leq i \leq p$, then $\Exp{f(S_o)} \geq \alpha b\tau (1 - p^{-1} - \alpha b / (1 + \alpha)) + \alpha(1 - b) \cdot f(\opt \setminus O)$.
\end{lemma}
\begin{proof}
Fix some integer $1 \leq i \leq p$, and consider the vector $(b \onevect{\hat{S}_i}) \vee \onevect{\opt \setminus O}$. Clearly,
\[
	\onenorm{(b \onevect{\hat{S}_i}) \vee \onevect{\opt \setminus O}}
	\leq
	b \cdot |\hat{S}_i| + |\opt \setminus O|
	\leq
	|O| + |\opt \setminus O|
	=
	|\opt|
	\leq
	k
	\enspace,
\]
where the second inequality holds by the definition of $b$ since $|\hat{S}_i| < k$. Given the last property of the vector $(b \onevect{\hat{S}_i}) \vee \onevect{\opt \setminus O}$, standard rounding techniques such as Pipage Rounding~\cite{DBLP:journals/siamcomp/CalinescuCPV11} and Swap Rounding~\cite{DBLP:conf/focs/ChekuriVZ10} can be used to produce from this vector a set $A_i \subseteq \hat{S}_i \cup (\opt \setminus O) \subseteq \bigcup_{i = 1}^p \hat{S}_i$ of size at most $k$ such that $f(A_i) \geq F((b \onevect{\hat{S}_i}) \vee \onevect{\opt \setminus O})$. Since $S_o$ is produced by the algorithm $\offlinealg$ whose approximation ratio is $\alpha$ and $A_i$ is one possible output for this algorithm, this implies $\Exp{f(S_o)} \geq \alpha \cdot f(A_i) \geq \alpha \cdot F((b \onevect{\hat{S}_i}) \vee \onevect{\opt \setminus O})$. Furthermore, by averaging this equation over all the possible values of $i$, we get
\begin{align*}
	\Exp{f(S_o)}
	\geq{} &
	\alpha p^{-1} \cdot \sum_{i = 1}^p F((b \onevect{\hat{S}_i}) \vee \onevect{\opt \setminus O})\\
	\geq{} &
	\alpha p^{-1} \cdot \left[(1 - b) \cdot \sum_{i = 1}^p f(\opt \setminus O) + b \cdot \sum_{i = 1}^p f(\hat{S}_i \cup (\opt \setminus O)) \right]\\
	\geq{} &
	\alpha(1 - b) \cdot f(\opt \setminus O) + \alpha b[(1 - p^{-1}) \cdot f(\opt) - \alpha b \tau / (1 + \alpha)]
	\enspace,
\end{align*}
where the second inequality holds since the submodularity of $f$ implies that $F((t \cdot \onevect{\hat{S}_i}) \vee \onevect{\opt \setminus O})$ is a concave function of $t$ within the range $[0, 1]$ (and $b$ falls is inside this range), and the last inequality holds by Lemma~\ref{lem:average} (for $A = \opt \setminus O$). The lemma now follows by recalling that $f(OPT) \geq \tau$.
\end{proof}

Combining the guarantees of the last two lemmata we can now obtain a lower bound on the value of the solution of Algorithm~\ref{alg:main_algorithm_deterministic} in the case that $|S_i| < k$ for every integer $1 \leq i \leq p$ which is independent of $b$.
\begin{corollary}
\label{cor:belowK_conclusionMain}
If $|\hat{S}_i| < k$ for every integer $1 \leq i \leq p$, then the output of Algorithm~\ref{alg:main_algorithm_deterministic} has value of at least $\left(\frac{\alpha}{\alpha + 1} - 2p^{-1}\right)\tau$.
\end{corollary}
\begin{proof}
The corollary follows immediately from the non-negativity of $f$ when $p = 1$. Thus, we may assume $p \geq 2$ in the rest of the proof.

Since Algorithm~\ref{alg:main_algorithm_deterministic} outputs the best solution among the $p + 1$ solutions it creates, we get by Lemma~\ref{lem:average} (for $A = \varnothing$) and Lemma~\ref{lem:belowKMain} that the value of the solution it outputs is at least
\begin{align*}
		\max\left\{p^{-1} \sum_{i = 1}^p f(\hat{S}_i), f(S_o)\right\}\mspace{-81mu}&\mspace{81mu}
		\geq
		\max\{(1 - p^{-1}) \cdot f(O) - \alpha b \tau / (1 + \alpha), \\&\alpha b\tau (1 - p^{-1} - \alpha b / (1 + \alpha)) + \alpha(1 - b) \cdot f(\opt \setminus O)\}\\
		\geq{} &
		\frac{\alpha(1 - b)}{\alpha(1 - b)+1-p^{-1}}\cdot
		\bracks{(1 - p^{-1}) \cdot f(O) - \alpha b \tau / (1 + \alpha)}\\
		&+\frac{1-p^{-1}}{\alpha(1 - b)+1-p^{-1}}\cdot \bracks{\alpha b\tau (1 - p^{-1} - \alpha b / (1 + \alpha)) + \alpha(1 - b) \cdot f(\opt \setminus O)}
		\enspace.
	\end{align*}
	Note now that the submodularity and non-negativity of $f$ guarantee together $f(O) + f(\opt \setminus O) \geq f(\opt) \geq \tau$. Using this fact and the non-negativity of $f$, the previous inequality yields
	\begin{align*}
		\frac{\max\left\{p^{-1} \sum_{i = 1}^p f(\hat{S}_i), f(S_o)\right\}}{\alpha\tau}
		\geq\mspace{-140mu}{} &\mspace{140mu}\\
		&\max\left\{0, \frac{(1 - b)(1 - p^{-1}) - \alpha b(1 - b)/(1 + \alpha) + b(1 - p^{-1})(1 - p^{-1} - \alpha b/(1 + \alpha))}{\alpha(1 - b)+1-p^{-1}}\right\}\\
		\geq{} &
		\max\left\{0, \frac{(1 - b)(1 - p^{-1}) - \alpha b(1 - b)/(1 + \alpha) + b(1 - p^{-1})(1 - p^{-1} - \alpha b/(1 + \alpha))}{\alpha(1 - b)+1}\right\}\\
		\geq{} &
		\frac{(1 - b) - \alpha b(1 - b)/(1 + \alpha) + b(1 - \alpha b/(1 + \alpha))}{\alpha(1 - b)+1} - 2p^{-1}
		\enspace,
	\end{align*}
	where the last two inequalities hold since $\alpha \in (0, 1]$, $b \in [0, 1]$ and $p \geq 2$. Simplifying the last inequality, we get
	\begin{align*}
		\frac{\max\left\{p^{-1} \sum_{i = 1}^p f(\hat{S}_i), f(S_o)\right\}}{\alpha\tau}
		\geq{} &
		\frac{(1 - b)(1 + \alpha) - \alpha b(1 - b) + b(1 + \alpha - \alpha b)}{(1 + \alpha)[\alpha(1 - b)+1]} - 2p^{-1}\nonumber\\
		={} &
		\frac{1 + \alpha(1 - b)}{(1 + \alpha)[\alpha(1 - b)+1]} - 2p^{-1}
		=
		\frac{1}{1 + \alpha} - 2p^{-1}
		\enspace.
	\end{align*}
	The corollary now follows by rearranging this inequality (and recalling that $\alpha \in (0, 1]$).
\end{proof}

Note that the guarantee of Corollary~\ref{cor:belowK_conclusionMain} (for the case it considers) is always weaker than the guarantee of Lemma~\ref{lem:Fx-exactlyKMain}. Thus, we can summarize the results we have proved so far using the following proposition.
\begin{proposition}
\label{prp:resultsMain}
Algorithm~\ref{alg:main_algorithm_deterministic} stores $O\left(pk\right)$ elements, makes at most $p$ marginal value calculations while processing each arriving element and its output set has an expected value of at least $\left(\frac{\alpha}{\alpha + 1} - 2p^{-1}\right)\tau$.
\end{proposition}

Using the last proposition, we can now prove the following theorem. As discussed at the beginning of the section, in Section~\ref{sec:full} we explain how the assumption that $\tau$ is known can be dropped at the cost of a slight increase in the number of of elements stored by the algorithm and its update time, which yields Theorem~\ref{trm:main}.
\begin{theorem}
\label{trm:non-polynomial-knowing-tauMain}
Assume there exists an $\alpha$-approximation offline algorithm $\offlinealg$ for maximizing a non-negative submodular function subject to cardinality constraint whose space complexity is nearly linear in the size of the ground set. Then, for every constant $\varepsilon \in (0, 1]$, there exists a semi-streaming algorithm that assumes access to an estimate $\tau$ of $f(\opt)$ obeying $(1-\varepsilon/2) \cdot f(\opt)\leq\tau\leq{f(\opt)}$ and provides $(\frac{\alpha}{1 + \alpha} - \varepsilon)$-approximation for the problem of maximizing a non-negative submodular function subject to cardinality constraint. This algorithm stores $O(k\varepsilon^{-1})$ elements and calculates $O(\varepsilon^{-1})$ marginal values while processing each arriving element.
\end{theorem}
\begin{proof}
Consider the algorithm obtained from Algorithm~\ref{alg:main_algorithm_deterministic} by setting $p = \lceil 4/\varepsilon \rceil$. By Proposition~\ref{prp:resultsMain} and the non-negativity of $f$, this algorithm stores only $O(pk) = O(k\varepsilon^{-1})$ elements, calculates only $p = O(\varepsilon^{-1})$ marginal values while processing each arriving element, and the expected value of its output set is at least
\[
	\max\left\{0, \left(\frac{\alpha}{\alpha + 1} - 2p^{-1}\right)\tau\right\}
	\geq
	\left(\frac{\alpha}{\alpha + 1} - \frac{\varepsilon}{2}\right) \cdot (1 - \varepsilon/2) \cdot f(\opt)
	\geq
	\left(\frac{\alpha}{\alpha + 1} - \varepsilon\right)\cdot f(\opt)
	\enspace,
\]
where the first inequality holds since $p \geq 4/\varepsilon$ and $\tau$ obeys, by assumption, $\tau \geq (1 - \varepsilon/2) \cdot f(\opt)$.
\end{proof}

\section{Complete Algorithm}\label{sec:full}
In this section, we explain how one can drop the assumption from Section~\ref{sec:simplified} that the algorithm has access to an estimate $\tau$ of $f(\opt)$. This completes the proof of Theorem~\ref{trm:main}.

Technically, we analyze in this section the variant of Algorithm~\ref{alg:main_algorithm_deterministic} given as Algorithm~\ref{alg:estimatingTau}. It gets the same two parameters $\alpha$ and $p$ received by Algorithm~\ref{alg:main_algorithm_deterministic} plus an additional parameter $\varepsilon' \in (0, 1)$ controlling the guaranteed quality of the output. The algorithm is based on a technique originally due to Kazemi et al.~\cite{DBLP:conf/icml/0001MZLK19}. Throughout its execution, Algorithm~\ref{alg:estimatingTau} tracks in $m$ a lower bound on the value of $f(\opt)$. The algorithm also maintains a set $T=\{(1+\varepsilon')^i\mid{m/ (1 + \varepsilon') \leq(1+\varepsilon')^i\leq{mk/c}}\}$ of values that, given the current value of the lower bound $m$, are (1) possible estimates for $f(\opt)$ at the current point and (2) are not so large that they dwarf this lower bound (of course, $T$ includes only a subset of the possible estimates obeying these requirements). For every estimate $\tau$ in $T$, the algorithm maintains $p$ sets $S^\tau_1, S^\tau_2, \dotsc, S^\tau_p$. We note that the set of solutions maintained is updated every time that $T$ is updated (which happens after every update of $m$). Specifically, whenever a new value $\tau$ is added to $T$, the algorithm instantiate $p$ new sets $S^\tau_1, S^\tau_2, \dotsc, S^\tau_p$, and whenever a value $\tau$ is dropped from $T$, the algorithm deletes $S^\tau_1, S^\tau_2, \dotsc, S^\tau_p$.

While a value $\tau$ remains in $T$, Algorithm~\ref{alg:estimatingTau} maintains the sets $S^\tau_1, S^\tau_2, \dotsc, S^\tau_p$ in exactly the same way that Algorithm~\ref{alg:main_algorithm_deterministic} maintains its sets $S_1, S_2, \dotsc, S_p$ given this value $\tau$ as an estimate for $f(\opt)$. Moreover, we show below that if $\tau$ remains in $T$ when the algorithm terminates, then the contents of $S^\tau_1, S^\tau_2, \dotsc, S^\tau_p$ when the algorithm terminates are equal to the contents of the sets $S_1, S_2, \dotsc, S_p$ when Algorithm~\ref{alg:main_algorithm_deterministic} terminates after executing with this $\tau$ as the estimate for $f(\opt)$. Thus, one can view Algorithm~\ref{alg:estimatingTau} as parallel execution of Algorithm~\ref{alg:main_algorithm_deterministic} for many estimates of $f(\opt)$ at the same time. After viewing the last element, Algorithm~\ref{alg:estimatingTau} calculates for every $\tau \in T$ an output set $\bar{S}_\tau$ based on the sets $S^\tau_1, S^\tau_2, \dotsc, S^\tau_p$ in the same way Algorithm~\ref{alg:main_algorithm_deterministic} does that, and then outputs the best output set computed for any $\tau \in T$.

\begin{algorithm2e}[th]
\caption{\streamdeterministic $(p, \alpha, \varepsilon')$} \label{alg:estimatingTau}
	\DontPrintSemicolon
	Let $c\gets\frac{\alpha}{1 + \alpha}$.\\
	Let $m\gets f(\varnothing)$ and $T\gets\curly{(1+\varepsilon')^h\mid{m / (1 + \varepsilon') \leq(1+\varepsilon')^h\leq{mk/c}}}$.\\
	\For {each arriving element $e$} {
		Let $m' \gets \max\{f(\curly{e}), \max_{\tau \in T, 1 \leq i \leq p} f(S_i^\tau)\}$.\\
		\If {$m<m'$} {
			Update $m\gets m'$, and then $T\gets\curly{(1+\varepsilon')^h\mid{m/ (1 + \varepsilon') \leq(1+\varepsilon')^h\leq{mk/c}}}$. \label{line:T_update_tau}\\
			Delete $S^\tau_1, S^\tau_2, \dotsc, S^\tau_p$ for every value $\tau$ removed from $T$ in Line~\ref{line:T_update_tau}.\\
			Initialize $S^\tau_i \gets \varnothing$ for every value $\tau$ added to $T$ in Line~\ref{line:T_update_tau} and integer $1 \leq i \leq p$.
		}
		\For {every $\tau\in{T}$} {
			\If{there exists an integer $1 \leq i \leq p$ such that $|S^\tau_i| < k$ and $f(e \mid S^\tau_i) \geq \frac{c\tau}{k}$}
			{
				Update $S^\tau_i \gets S^\tau_i \cup \{e\}$ (if there are multiple options for $i$, pick an arbitrary one).
			}
		}
	}
	\For {every $\tau\in{T}$} {
		Find another feasible solution $S^\tau_o \subseteq \bigcup_{i = 1}^p S^\tau_i$ by running $\offlinealg$ with $\bigcup_{i = 1}^p S^\tau_i$ as the ground set.\\
		Let $\bar{S}_\tau$ be the better solution among $S^\tau_o$ and the $p$ solutions $S^\tau_1, S^\tau_2, \dotsc, S^\tau_p$.
	}
	\Return{the best solution among $\{\bar{S}_\tau\}_{\tau \in T}$, or the empty set if $T = \varnothing$}.
\end{algorithm2e}

We begin the analysis of Algorithm~\ref{alg:estimatingTau} by providing a basic lower bound on the value of each set $S^\tau_i$. This lower bound can be viewed as a generalized counterpart of Lemma~\ref{lem:Fx-exactlyKMain}.

\begin{lemma}
\label{lem:Fx-exactlyK_no_assumption}
At every point during the execution of the algorithm, for every $\tau \in T$ and $1 \leq i \leq p$ it holds that $f(S^\tau_i) \geq \frac{\alpha\tau \cdot |S^\tau_i|}{(1 + \alpha)k}$.
\end{lemma}
\begin{proof}
Denote by $e_1,e_2,\dots,e_{|S^\tau_i|}$ the elements of $S^\tau_i$ in the order of their arrival. Using this notation, the value of $f(S^\tau_i)$ can be written as follows.
	\[
		f(S^\tau_i)
		=
		f(\varnothing)+\sum_{j=1}^{|S^\tau_i|}{f\big(e_j \mid \curly{e_1,e_2,\dotsc,e_{j-1}}\big)}
		\geq
		\sum_{i=1}^{|S^\tau_i|} \frac{\alpha\tau}{(1 + \alpha)k}
		=
		\frac{\alpha\tau \cdot |S^\tau_i|}{(1 + \alpha)k}
		\enspace,
	\]
where the inequality holds since the non-negativity of $f$ implies $f(\varnothing) \geq 0$ and Algorithm~\ref{alg:estimatingTau} adds an element $e_j$ to $S^\tau_i$ only when $f\big(e_j \mid \curly{e_1,e_2,\dotsc,e_{j-1}}\big)\geq\frac{c\tau}{k} = \frac{\alpha\tau}{(1 + \alpha)k}$.
\end{proof}

Using the last lemma we can now prove the following observation, which upper bounds the size of each set $S^\tau_i$ maintained by Algorithm~\ref{alg:estimatingTau}, and thus, serves as a first step towards bounding the space complexity of this algorithm.
\begin{observation} \label{obs:supp_size}
At the end of the every iteration of Algorithm~\ref{alg:estimatingTau}, the size of the set $S^\tau_i$ is at most $mk/ (c \tau) + 1$ for every $\tau \in T$ and integer $1 \leq i \leq p$.
\end{observation}
\begin{proof}
Let $\bar{S}^\tau_i$ denote the content of the set $S^\tau_i$ at the beginning of the iteration. Since at most a single element is added to $S^\tau_i$ during the iteration, we get via Lemma~\ref{lem:Fx-exactlyK_no_assumption} that the value of the set $\bar{S}^\tau_i$ is at least
\[
	\frac{\alpha\tau \cdot |\bar{S}^\tau_i|}{(1 + \alpha)k}
	\geq
	\frac{\alpha\tau \cdot (|S^\tau_i| - 1)}{(1 + \alpha)k}
	=
	\frac{c\tau \cdot (|S^\tau_i| - 1)}{k}
	\enspace.
\]
The way in which Algorithm~\ref{alg:estimatingTau} updates $m$ guarantees that immediately after the update of $m$ at the beginning of the iteration, the value of $m$ was at least as large as the value of $\bar{S}^\tau_i$, and thus, we get
\[
	m
	\geq
	\frac{c\tau \cdot (|S^\tau_i| - 1)}{k}
	\enspace,
\]
which implies the observation.
\end{proof}

We are now ready to bound the space complexity and update time of Algorithm~\ref{alg:estimatingTau}.
\begin{lemma} \label{lem:space_tau}
Algorithm~\ref{alg:estimatingTau} can be implemented so that it stores $O(pk(\varepsilon')^{-1} \log \alpha^{-1})$ elements and makes only $O(p(\varepsilon')^{-1}\log (k / \alpha))$ marginal value calculations when processing each arriving element.
\end{lemma}
\begin{proof}
We begin the proof with two technical calculations. First, note that the number of estimates in $T$ is upper bounded at all times by
\[
	1 + \log_{1 + \varepsilon'} \left(\frac{km/c}{m/(1 + \varepsilon')}\right)
	=
	2 + \frac{\ln k - \ln c}{\ln(1 + \varepsilon')}
	\leq
	2 + \frac{\ln k - \ln c}{2\varepsilon'/3}
	=
	O((\varepsilon')^{-1}(\log k - \log c))
	\enspace.
\]
Second, note that
\begin{align*}
	\sum_{\tau \in T} \sum_{i = 1}^p (|S^\tau_i| - 1) 
	\leq{} &
	\sum_{\tau \in T} \min\left\{pk, \frac{pmk}{c \tau}\right\}
	\leq
	\mspace{-18mu} \sum_{\substack{h \in \mathbb{Z} \\ (1 + \varepsilon')^{h + 1} \geq m}} \mspace{
-18mu} \min\left\{pk, \frac{pmk}{c (1 + \varepsilon')^h}\right\}\\
	\leq{} &
	pk \cdot \left|\left\{h \in \mathbb{Z} \mid \frac{m}{1 + \varepsilon'} \leq (1 + \varepsilon')^{h} < m / c\right\}\right|
	+
	\sum_{h = 0}^\infty \frac{pk}{(1 + \varepsilon')^h}\\
	\leq{} &
	pk \cdot \left(\log_{1 + \varepsilon'} \left(\frac{m/c}{m}\right) + 2\right) +
	\frac{pk}{1 - 1/(1 + \varepsilon')}
	\leq
	\frac{pk \cdot (\ln c^{-1} + 2)}{\ln (1 + \varepsilon')} + 
	\frac{pk(1 + \varepsilon')}{\varepsilon'}\\
	={} &
	O(pk(\varepsilon')^{-1}\log c^{-1}) + O(pk(\varepsilon')^{-1})
	=
	O(pk(\varepsilon')^{-1}\log c^{-1})
	\enspace,
\end{align*}
where the first inequality follows from Observation~\ref{obs:supp_size} and the fact that all the sets maintained by Algorithm~\ref{alg:estimatingTau} are of size at most $k$; and the second inequality follows from the definition of $T$.

We now observe that, while processing each arriving element, Algorithm~\ref{alg:estimatingTau} makes $p$ marginal value calculations in association with every value $\tau \in T$ and another set of $O(p)$ such calculations that are not associated with any value of $T$. Thus, the total number of marginal value calculations made by the algorithm during the processing of an arriving element is
\[
	p \cdot |T| + O(p)
	=
	p \cdot O((\varepsilon')^{-1}(\log k - \log c)) + O(p)
	=
	O(p(\varepsilon')^{-1}(\log k - \log c))
	=
	O(p(\varepsilon')^{-1}\log (k/c))
	\enspace.
\]
We also note that Algorithm~\ref{alg:estimatingTau} has to store only the elements belonging to $S^\tau_i$ for some $\tau \in T$ and integer $1 \leq i \leq p$. Thus, in all these sets together the algorithm stores $O(pk(\varepsilon')^{-1}\log c^{-1})$ elements since
\begin{align*}
	\sum_{\tau \in T} \sum_{i = 1}^p |S^\tau_i|
	={} &
	\sum_{\tau \in T} \sum_{i = 1}^p (|\supp(S^\tau_i)| - 1) + p|T|\\
	={} &
	O(pk(\varepsilon')^{-1}\log c^{-1}) + p \cdot O((\varepsilon')^{-1}(\log k - \log c))
	=
	O(pk(\varepsilon')^{-1}\log c^{-1})
	\enspace.
\end{align*}
To complete the proof of the lemma, it remains to note that $c = \alpha / (1 + \alpha) \geq \alpha / 2$.
\end{proof}

A this point we divert our attention to analyzing the approximation gaurantee of Algorithm~\ref{alg:estimatingTau}. Let us denote by $\hat{m}$ and $\hat{T}$ the final values of $m$ and $T$, respectively. 
\begin{observation}
\label{obs:T-not-empty}
$c \in (0, 1/2]$, and thus, $\hat{T}$ is not empty unless $\hat{m} = 0$.
\end{observation}
\begin{proof}
Since $\alpha \in (0, 1]$ and $\alpha / (\alpha + 1)$ is an increasing function of $\alpha$,
\[
	c
	=
	\frac{\alpha}{\alpha + 1}
	\in
	\left(\frac{0}{0 + 1}, \frac{1}{1 + 1}\right]
	=
	(0, 1/2]
	\enspace.
	\qedhere
\]
\end{proof}

The last observation immediately implies that when $\hat{T}$ is empty, all the singletons have zero values, and the same must be true for every other non-empty set by the submodularity and non-negativity of $f$. Thus, $\opt = \varnothing$, which makes the output of Algorithm~\ref{alg:estimatingTau} optimal in the rare case in which $\hat{T}$ is empty. Hence, we can assume from now on that $\hat{T} \neq \varnothing$. The following lemma shows that $\hat{T}$ contains a good estimate for $f(\opt)$ in this case.
\begin{lemma}
\label{lem:apxtau-tau-in-T}
	The set $\hat{T}$ contains a value $\hat{\tau}$ such that $(1-\varepsilon') \cdot f(\opt)\leq\hat{\tau}\leq f(\opt)$.
\end{lemma}
\begin{proof}
Observe that $\hat{m} \geq \max\curly{f(\varnothing), \max_{e\in\groundset}\curly{f\big(\curly{e}\big)}}$. Thus, by the submodularity of $f$,
\[
	f(\opt)
	\leq
	f(\varnothing)
	+
	\sum_{e \in \opt} \mspace{-9mu} \bracks{f(\{e\}) - f(\varnothing)}
	\leq
	\max\curly{f(\varnothing), \sum_{e \in \opt} \mspace{-9mu} f\big(\{e\}\big)}
	\leq
	k\hat{m}
	\leq
	\frac{k\hat{m}}{c}
	\enspace.
\]
In contrast, one can note that $m$ is equal to the value of $f$ for some feasible solution, and thus, $f(\opt) \geq \hat{m}$. 
Since $\hat{T}$ contains all the values of the form $(1 + \varepsilon')^i$ in the range $[\hat{m} / (1 + \varepsilon'), k\hat{m}/c]$, the above inequalities imply that it contains in particular the largest value of this form that is still not larger than $f(\opt)$. Let us denote this value by $\hat{\tau}$. By definition, $\hat{\tau} \leq f(\opt)$. Additionally,
\[
	\hat{\tau} \cdot (1 + \varepsilon') \geq f(\opt)
	\Rightarrow
	\hat{\tau} \geq \frac{f(\opt)}{1 + \varepsilon'} \geq (1 - \varepsilon') \cdot f(\opt)
	\enspace.
	\qedhere
\]
\end{proof}

Let us now concentrate on the value $\hat{\tau}$ whose existence is guaranteed by Lemma~\ref{lem:apxtau-tau-in-T}, and let $\hat{S}_1, \hat{S}_2, \dotsc, \hat{S}_p$ denote the sets maintained by Algorithm~\ref{alg:main_algorithm_deterministic} when it gets $\hat{\tau}$ as the estimate for $f(\opt)$. Additionally, let us denote by $e_1,e_2,\dots,e_n$ the elements of $\groundset$ in the order of their arrival, and let $e_j$ be the element whose arrival made Algorithm~\ref{alg:estimatingTau} add $\hat{\tau}$ to $T$ (i.e., $\hat{\tau}$ was added to $T$ by Algorithm~\ref{alg:estimatingTau} while processing $e_j$). If $\hat{\tau}$ belonged to $T$ from the very beginning of the execution of Algorithm~\ref{alg:estimatingTau}, then we define $j = 1$
\begin{lemma} \label{lem:apxtau-before-tau} 
All the sets $\hat{S}_1, \hat{S}_2, \dotsc, \hat{S}_p$ maintained by Algorithm~\ref{alg:main_algorithm_deterministic} are empty immediately prior to the arrival of $e_j$.
\end{lemma}
\begin{proof}
If $j = 1$, then lemma is trivial. Thus, we assume $j > 1$ in the rest of this proof. Prior to the arrival of $e_j$, $\hat{\tau}$ was not part of $T$. Nevertheless, since $\hat{\tau} \in \hat{T}$, we must have at all times
\[
	\hat{\tau}
	\geq
	\frac{\hat{m}}{1 + \varepsilon'}
	\geq
	\frac{m}{1 + \varepsilon'}
	\enspace,
\]
where the second inequality holds since the value $m$ only increases over time. Therefore, the reason that $\hat{\tau}$ did not belong to $T$ prior to the arrival of $e_j$ must have been that $m$ was smaller than $c\hat{\tau}/k$. Since $m$ is at least as large as the value of any singleton containing an element already viewed by the algorithm we get, for every two integers $1 \leq t \leq j - 1$ and $1 \leq i \leq p$,
	\[
		\frac{c\hat{\tau}}{k}
		>
		f(\{e_t\})
		\geq
		f(\{e_t\} \mid \varnothing)
		\geq
		f\big(e_t \mid \hat{S}_i \wedge \onevect{\curly{e_1,e_2,\dotsc,e_{t-1}}})
		\enspace,
	\]
	where the second inequality follows from the non-negativity of $f$ and the last from its submodularity. Thus, $e_t$ is not added by Algorithm~\ref{alg:main_algorithm_deterministic} to any one of the sets $\hat{S}_1, \hat{S}_2, \dotsc, \hat{S}_p$, which implies that all these sets are empty at the moment $e_j$ arrives.
\end{proof}

According to the above discussion, from the moment $\hat{\tau}$ gets into $T$, Algorithm~\ref{alg:estimatingTau} updates the sets $S^{\hat{\tau}}_1, S^{\hat{\tau}}_2, \dotsc, S^{\hat{\tau}}_p$ in the same way that Algorithm~\ref{alg:main_algorithm_deterministic} updates $\hat{S}_1, \hat{S}_2, \dotsc, \hat{S}_p$ (note that, once $\hat{\tau}$ gets into $T$, it remains there for good since $\hat{\tau} \in \hat{T}$). Together with the previous lemma which shows that $\hat{S}_1, \hat{S}_2, \dotsc, \hat{S}_p$ are empty just like $S^{\hat{\tau}}_1, S^{\hat{\tau}}_2, \dotsc, S^{\hat{\tau}}_p$ at the moment $e_j$ arrives, this implies that the final contents of $S^{\hat{\tau}}_1, S^{\hat{\tau}}_2, \dotsc, S^{\hat{\tau}}_p$ are equal to the final contents of $\hat{S}_1, \hat{S}_2, \dotsc, \hat{S}_p$, respectively. Since the set $\bar{S}_{\hat{\tau}}$ is computed based on the final contents of $S^{\hat{\tau}}_1, S^{\hat{\tau}}_2, \dotsc, S^{\hat{\tau}}_p$ in the same way that the output of Algorithm~\ref{alg:main_algorithm_deterministic} is computed based on the final contents of $\hat{S}_1, \hat{S}_2, \dotsc, \hat{S}_p$, we get the following corollary.

\begin{corollary} \label{cor:reduction1-2}
If it is guaranteed that the approximation ratio of Algorithm~\ref{alg:main_algorithm_deterministic} is at least $\beta$ when $(1 - \varepsilon') \cdot f(\opt) \leq \tau \leq f(\opt)$ for some choice of the parameters $\alpha$ and $p$, then the approximation ratio of Algorithm~\ref{alg:estimatingTau} is at least $\beta$ as well for this choice of $\alpha$ and $p$.
\end{corollary}

We are now ready to prove Theorem~\ref{trm:main}.
\begin{reptheorem}{trm:main}
\TrmMain
\end{reptheorem}
\begin{proof}
The proof of Theorem~\ref{trm:non-polynomial-knowing-tauMain} shows that Algorithm~\ref{alg:main_algorithm_deterministic} achieves an approximation guarantee of $\frac{\alpha}{1 + \alpha} - \varepsilon$ when it has access to a value $\tau$ obeying $(1-\varepsilon/2) \cdot f(\opt)\leq\tau\leq{f(\opt)}$ and its parameter $p$ is set to $\lceil 4 / \varepsilon \rceil$. According to Corollary~\ref{cor:reduction1-2}, this implies that by setting the parameter $p$ of Algorithm~\ref{alg:estimatingTau} in the same way and setting $\varepsilon'$ to $\varepsilon / 2$, we get an algorithm whose approximation ratio is at least $\frac{\alpha}{1 + \alpha} - \varepsilon$ and does not assume access to an estimate $\tau$ of $f(\opt)$.

It remains to bound the space requirement and update time of the algorithm obtained in this way. Plugging the equalities $p = \lceil 4 / \varepsilon \rceil$ and $\varepsilon' = \varepsilon / 2$ into the guarantee of Lemma~\ref{lem:space_tau}, we get that the algorithm we have obtained stores $O(k\varepsilon^{-2} \log \alpha^{-1})$ elements and makes $O(\varepsilon^{-2} \log(k/\alpha))$ marginal value calculations while processing each arriving element.
\end{proof}

\newpage{}

\bibliographystyle{plainurl}
\bibliography{streaming-nmc-arxiv}

\begin{thebibliography}{10}

\bibitem{DBLP:conf/icalp/AlalufEFNS20}
Naor Alaluf, Alina Ene, Moran Feldman, Huy~L. Nguyen, and Andrew Suh.
\newblock Optimal streaming algorithms for submodular maximization with
  cardinality constraints.
\newblock In {\em ICALP}, pages 6:1--6:19, 2020.
\newblock URL: \url{https://doi.org/10.4230/LIPIcs.ICALP.2020.6}, \href
  {http://dx.doi.org/10.4230/LIPIcs.ICALP.2020.6}
  {\path{doi:10.4230/LIPIcs.ICALP.2020.6}}.

\bibitem{oldArxivVersion}
Naor Alaluf, Alina Ene, Moran Feldman, Huy~L. Nguyen, and Andrew Suh.
\newblock Optimal streaming algorithms for submodular maximization with
  cardinality constraints.
\newblock {\em CoRR}, abs/1911.12959v2, 2020.
\newblock URL: \url{https://arxiv.org/pdf/1911.12959v2.pdf}.

\bibitem{DBLP:conf/esa/AzarGR11}
Yossi Azar, Iftah Gamzu, and Ran Roth.
\newblock Submodular max-sat.
\newblock In {\em ESA}, pages 323--334, 2011.
\newblock URL: \url{https://doi.org/10.1007/978-3-642-23719-5\_28}, \href
  {http://dx.doi.org/10.1007/978-3-642-23719-5\_28}
  {\path{doi:10.1007/978-3-642-23719-5\_28}}.

\bibitem{badanidiyuru2014}
Ashwinkumar Badanidiyuru, Baharan Mirzasoleiman, Amin Karbasi, and Andreas
  Krause.
\newblock Streaming submodular maximization: Massive data summarization on the
  fly.
\newblock In {\em Proceedings of the 20th ACM SIGKDD international conference
  on Knowledge discovery and data mining}, pages 671--680. ACM, 2014.

\bibitem{Barbosa2016}
Rafael da~Ponte Barbosa, Alina Ene, Huy~L Nguyen, and Justin Ward.
\newblock A new framework for distributed submodular maximization.
\newblock In {\em IEEE Foundations of Computer Science (FOCS)}, pages 645--654,
  2016.

\bibitem{Barbosa2015}
Rafael~D.P. Barbosa, Alina Ene, Huy~L. Nguyen, and Justin Ward.
\newblock The power of randomization: Distributed submodular maximization on
  massive datasets.
\newblock In {\em International Conference on Machine Learning (ICML)}, 2015.

\bibitem{DBLP:journals/talg/BuchbinderF18}
Niv Buchbinder and Moran Feldman.
\newblock Deterministic algorithms for submodular maximization problems.
\newblock {\em {ACM} Trans. Algorithms}, 14(3):32:1--32:20, 2018.
\newblock URL: \url{https://doi.org/10.1145/3184990}, \href
  {http://dx.doi.org/10.1145/3184990} {\path{doi:10.1145/3184990}}.

\bibitem{buchbinder2019constrained}
Niv Buchbinder and Moran Feldman.
\newblock Constrained submodular maximization via a non-symmetric technique.
\newblock {\em Mathematics of Operations Research}, 44(3):988--1005, 2019.

\bibitem{BuchbinderFG19}
Niv Buchbinder, Moran Feldman, Yuval Filmus, and Mohit Garg.
\newblock Online submodular maximization: Beating 1/2 made simple.
\newblock In {\em Integer Programming and Combinatorial Optimization - 20th
  International Conference, {IPCO} 2019, Ann Arbor, MI, USA, May 22-24, 2019,
  Proceedings}, pages 101--114, 2019.
\newblock URL: \url{https://doi.org/10.1007/978-3-030-17953-3\_8}, \href
  {http://dx.doi.org/10.1007/978-3-030-17953-3\_8}
  {\path{doi:10.1007/978-3-030-17953-3\_8}}.

\bibitem{DBLP:conf/soda/BuchbinderF019}
Niv Buchbinder, Moran Feldman, and Mohit Garg.
\newblock Deterministic ($\nicefrac{1}{2} + \varepsilon$)-approximation for
  submodular maximization over a matroid.
\newblock In {\em SODA}, pages 241--254, 2019.
\newblock URL: \url{https://doi.org/10.1137/1.9781611975482.16}, \href
  {http://dx.doi.org/10.1137/1.9781611975482.16}
  {\path{doi:10.1137/1.9781611975482.16}}.

\bibitem{DBLP:conf/soda/BuchbinderFNS14}
Niv Buchbinder, Moran Feldman, Joseph Naor, and Roy Schwartz.
\newblock Submodular maximization with cardinality constraints.
\newblock In {\em SODA}, pages 1433--1452, 2014.
\newblock URL: \url{https://doi.org/10.1137/1.9781611973402.106}, \href
  {http://dx.doi.org/10.1137/1.9781611973402.106}
  {\path{doi:10.1137/1.9781611973402.106}}.

\bibitem{BuchbinderFS19}
Niv Buchbinder, Moran Feldman, and Roy Schwartz.
\newblock Online submodular maximization with preemption.
\newblock {\em {ACM} Trans. Algorithms}, 15(3):30:1--30:31, 2019.
\newblock URL: \url{https://doi.org/10.1145/3309764}, \href
  {http://dx.doi.org/10.1145/3309764} {\path{doi:10.1145/3309764}}.

\bibitem{DBLP:journals/siamcomp/CalinescuCPV11}
Gruia C{\u{a}}linescu, Chandra Chekuri, Martin P{\'{a}}l, and Jan
  Vondr{\'{a}}k.
\newblock Maximizing a monotone submodular function subject to a matroid
  constraint.
\newblock {\em {SIAM} J. Comput.}, 40(6):1740--1766, 2011.
\newblock URL: \url{https://doi.org/10.1137/080733991}, \href
  {http://dx.doi.org/10.1137/080733991} {\path{doi:10.1137/080733991}}.

\bibitem{chakrabarti2015submodular}
Amit Chakrabarti and Sagar Kale.
\newblock Submodular maximization meets streaming: Matchings, matroids, and
  more.
\newblock {\em Mathematical Programming}, 154(1-2):225--247, 2015.

\bibitem{DBLP:journals/talg/ChanHJKT18}
T.{-}H.~Hubert Chan, Zhiyi Huang, Shaofeng~H.{-}C. Jiang, Ning Kang, and
  Zhihao~Gavin Tang.
\newblock Online submodular maximization with free disposal.
\newblock {\em {ACM} Trans. Algorithms}, 14(4):56:1--56:29, 2018.
\newblock URL: \url{https://doi.org/10.1145/3242770}, \href
  {http://dx.doi.org/10.1145/3242770} {\path{doi:10.1145/3242770}}.

\bibitem{communication:Chekuri18}
Chandra Chekuri.
\newblock Personal communication, 2018.

\bibitem{chekuri2015streaming}
Chandra Chekuri, Shalmoli Gupta, and Kent Quanrud.
\newblock Streaming algorithms for submodular function maximization.
\newblock In {\em International Colloquium on Automata, Languages and
  Programming (ICALP)}, pages 318--330. Springer, 2015.

\bibitem{DBLP:conf/focs/ChekuriVZ10}
Chandra Chekuri, Jan Vondr{\'{a}}k, and Rico Zenklusen.
\newblock Dependent randomized rounding via exchange properties of
  combinatorial structures.
\newblock In {\em FOCS}, pages 575--584, 2010.
\newblock URL: \url{https://doi.org/10.1109/FOCS.2010.60}, \href
  {http://dx.doi.org/10.1109/FOCS.2010.60} {\path{doi:10.1109/FOCS.2010.60}}.

\bibitem{DBLP:journals/siamcomp/ChekuriVZ14}
Chandra Chekuri, Jan Vondr{\'{a}}k, and Rico Zenklusen.
\newblock Submodular function maximization via the multilinear relaxation and
  contention resolution schemes.
\newblock {\em {SIAM} J. Comput.}, 43(6):1831--1879, 2014.
\newblock URL: \url{https://doi.org/10.1137/110839655}, \href
  {http://dx.doi.org/10.1137/110839655} {\path{doi:10.1137/110839655}}.

\bibitem{ene2016constrained}
Alina Ene and Huy~L Nguyen.
\newblock Constrained submodular maximization: Beyond 1/e.
\newblock In {\em IEEE Foundations of Computer Science (FOCS)}, pages 248--257.
  IEEE, 2016.

\bibitem{epasto2017bicriteria}
Alessandro Epasto, Vahab Mirrokni, and Morteza Zadimoghaddam.
\newblock Bicriteria distributed submodular maximization in a few rounds.
\newblock In {\em PACM Symposium on Parallelism in Algorithms and Architectures
  (SPAA)}, pages 25--33, 2017.

\bibitem{feige2011maximizing}
Uriel Feige, Vahab~S Mirrokni, and Jan Vondr{\'a}k.
\newblock Maximizing non-monotone submodular functions.
\newblock {\em SIAM Journal on Computing}, 40(4):1133--1153, 2011.

\bibitem{DBLP:journals/talg/Feldman17}
Moran Feldman.
\newblock Maximizing symmetric submodular functions.
\newblock {\em {ACM} Trans. Algorithms}, 13(3):39:1--39:36, 2017.
\newblock URL: \url{https://doi.org/10.1145/3070685}, \href
  {http://dx.doi.org/10.1145/3070685} {\path{doi:10.1145/3070685}}.

\bibitem{DBLP:conf/colt/FeldmanHK17}
Moran Feldman, Christopher Harshaw, and Amin Karbasi.
\newblock Greed is good: Near-optimal submodular maximization via greedy
  optimization.
\newblock In {\em Proceedings of the 30th Conference on Learning Theory, {COLT}
  2017, Amsterdam, The Netherlands, 7-10 July 2017}, pages 758--784, 2017.
\newblock URL: \url{http://proceedings.mlr.press/v65/feldman17b.html}.

\bibitem{feldman2018less}
Moran Feldman, Amin Karbasi, and Ehsan Kazemi.
\newblock Do less, get more: Streaming submodular maximization with
  subsampling.
\newblock In {\em Advances in Neural Information Processing Systems (NeurIPS)},
  pages 732--742, 2018.

\bibitem{Feldman2011a}
Moran Feldman, Joseph~(Seffi) Naor, and Roy Schwartz.
\newblock A unified continuous greedy algorithm for submodular maximization.
\newblock In {\em Proceedings of the 52nd IEEE Foundations of Computer Science
  (FOCS)}, pages 570--579. IEEE Computer Society, 2011.

\bibitem{feldman2020oneway}
Moran Feldman, Ashkan Norouzi-Fard, Ola Svensson, and Rico Zenklusen.
\newblock The one-way communication complexity of submodular maximization with
  applications to streaming and robustness.
\newblock In {\em STOC}, pages 1363--1374, 2020.

\bibitem{Fisher1978}
M.~L. Fisher, G.~L. Nemhauser, and L.~A. Wolsey.
\newblock An analysis of approximations for maximizing submodular set
  functions---{II}.
\newblock {\em Mathematical Programming Studies}, 8:73--87, 1978.

\bibitem{Gharan2011}
Shayan~Oveis Gharan and Jan Vondr\'{a}k.
\newblock Submodular maximization by simulated annealing.
\newblock In {\em ACM-SIAM Symposium on Discrete Algorithms (SODA)}, 2011.

\bibitem{DBLP:conf/soda/KapralovPV13}
Michael Kapralov, Ian Post, and Jan Vondr{\'{a}}k.
\newblock Online submodular welfare maximization: Greedy is optimal.
\newblock In {\em SODA}, pages 1216--1225, 2013.
\newblock URL: \url{https://doi.org/10.1137/1.9781611973105.88}, \href
  {http://dx.doi.org/10.1137/1.9781611973105.88}
  {\path{doi:10.1137/1.9781611973105.88}}.

\bibitem{DBLP:conf/icml/0001MZLK19}
Ehsan Kazemi, Marko Mitrovic, Morteza Zadimoghaddam, Silvio Lattanzi, and Amin
  Karbasi.
\newblock Submodular streaming in all its glory: Tight approximation, minimum
  memory and low adaptive complexity.
\newblock In {\em ICML}, pages 3311--3320, 2019.
\newblock URL: \url{http://proceedings.mlr.press/v97/kazemi19a.html}.

\bibitem{DBLP:journals/siamcomp/KorulaMZ18}
Nitish Korula, Vahab~S. Mirrokni, and Morteza Zadimoghaddam.
\newblock Online submodular welfare maximization: Greedy beats 1/2 in random
  order.
\newblock {\em {SIAM} J. Comput.}, 47(3):1056--1086, 2018.
\newblock URL: \url{https://doi.org/10.1137/15M1051142}, \href
  {http://dx.doi.org/10.1137/15M1051142} {\path{doi:10.1137/15M1051142}}.

\bibitem{Kumar2013}
Ravi Kumar, Benjamin Moseley, Sergei Vassilvitskii, and Andrea Vattani.
\newblock Fast greedy algorithms in mapreduce and streaming.
\newblock In {\em PACM Symposium on Parallelism in Algorithms and Architectures
  (SPAA)}, pages 1--10, 2013.

\bibitem{lee2009non}
Jon Lee, Vahab~S Mirrokni, Viswanath Nagarajan, and Maxim Sviridenko.
\newblock Non-monotone submodular maximization under matroid and knapsack
  constraints.
\newblock In {\em ACM Symposium on Theory of Computing (STOC)}, pages 323--332,
  2009.

\bibitem{lee2010submodular}
Jon Lee, Maxim Sviridenko, and Jan Vondr{\'a}k.
\newblock Submodular maximization over multiple matroids via generalized
  exchange properties.
\newblock {\em Mathematics of Operations Research}, 35(4):795--806, 2010.

\bibitem{DBLP:conf/ismp/Lovasz82}
L{\'{a}}szl{\'{o}} Lov{\'{a}}sz.
\newblock Submodular functions and convexity.
\newblock In {\em Mathematical Programming The State of the Art, XIth
  International Symposium on Mathematical Programming, Bonn, Germany, August
  23-27, 1982}, pages 235--257, 1982.
\newblock URL: \url{https://doi.org/10.1007/978-3-642-68874-4\_10}, \href
  {http://dx.doi.org/10.1007/978-3-642-68874-4\_10}
  {\path{doi:10.1007/978-3-642-68874-4\_10}}.

\bibitem{MirrokniZ15}
Vahab Mirrokni and Morteza Zadimoghaddam.
\newblock Randomized composable core-sets for distributed submodular
  maximization.
\newblock In {\em ACM Symposium on Theory of Computing (STOC)}, 2015.

\bibitem{mirzasoleiman2018streaming}
Baharan Mirzasoleiman, Stefanie Jegelka, and Andreas Krause.
\newblock Streaming non-monotone submodular maximization: Personalized video
  summarization on the fly.
\newblock In {\em Thirty-second AAAI conference on artificial intelligence},
  2018.

\bibitem{Mirzasoleiman2015distributed}
Baharan Mirzasoleiman, Amin Karbasi, Ashwinkumar Badanidiyuru, and Andreas
  Krause.
\newblock Distributed submodular cover: Succinctly summarizing massive data.
\newblock In {\em Advances in Neural Information Processing Systems}, pages
  2881--2889, 2015.

\bibitem{Mirzasoleiman2013}
Baharan Mirzasoleiman, Amin Karbasi, Rik Sarkar, and Andreas Krause.
\newblock Distributed submodular maximization: Identifying representative
  elements in massive data.
\newblock In {\em Advances in Neural Information Processing Systems (NeurIPS)},
  pages 2049--2057, 2013.

\bibitem{Nemhauser1978}
G.~L. Nemhauser, L.~A. Wolsey, and M.~L. Fisher.
\newblock An analysis of approximations for maximizing submodular set
  functions---{I}.
\newblock {\em Mathematical Programming}, 14(1):265--294, 1978.

\bibitem{DBLP:journals/mor/NemhauserW78}
George~L. Nemhauser and Laurence~A. Wolsey.
\newblock Best algorithms for approximating the maximum of a submodular set
  function.
\newblock {\em Math. Oper. Res.}, 3(3):177--188, 1978.
\newblock URL: \url{https://doi.org/10.1287/moor.3.3.177}, \href
  {http://dx.doi.org/10.1287/moor.3.3.177} {\path{doi:10.1287/moor.3.3.177}}.

\bibitem{DBLP:journals/mp/NemhauserWF78}
George~L. Nemhauser, Laurence~A. Wolsey, and Marshall~L. Fisher.
\newblock An analysis of approximations for maximizing submodular set
  functions---{I}.
\newblock {\em Math. Program.}, 14(1):265--294, 1978.
\newblock URL: \url{https://doi.org/10.1007/BF01588971}, \href
  {http://dx.doi.org/10.1007/BF01588971} {\path{doi:10.1007/BF01588971}}.

\bibitem{norouzi2018beyond}
Ashkan Norouzi-Fard, Jakub Tarnawski, Slobodan Mitrovic, Amir Zandieh,
  Aidasadat Mousavifar, and Ola Svensson.
\newblock Beyond 1/2-approximation for submodular maximization on massive data
  streams.
\newblock In {\em International Conference on Machine Learning}, pages
  3826--3835, 2018.

\bibitem{DBLP:journals/siamcomp/Vondrak13}
Jan Vondr{\'{a}}k.
\newblock Symmetry and approximability of submodular maximization problems.
\newblock {\em {SIAM} J. Comput.}, 42(1):265--304, 2013.
\newblock URL: \url{https://doi.org/10.1137/110832318}, \href
  {http://dx.doi.org/10.1137/110832318} {\path{doi:10.1137/110832318}}.

\end{thebibliography}

\newpage{}

\appendix
\section{Details about the Error in a Previous Work} \label{app:error}

As mentioned above, Chekuri et al.~\cite{chekuri2015streaming} described a semi-streaming algorithm for the problem of maximizing a non-negative (not necessarily monotone) submodular function subject to a cardinality constraint, and claimed an approximation ratio of roguhly $0.212$ for this algorithm. However, an error was later found in the proof of this result~\cite{communication:Chekuri18} (the error does not affect the other results of~\cite{chekuri2015streaming}). For completeness, we briefly describe in this appendix the error found.

In the proof presented by Chekuri et al.~\cite{chekuri2015streaming}, the output set of their algorithm is denoted by $\tilde{S}$. As is standard in the analysis of algorithms based on single-threshold Greedy, the analysis distinguishes between two cases: one case in which $\tilde{S} = k$, and a second case in which $\tilde{S} < k$. To argue about the second case, the analysis then implicitly uses the inequality
\begin{equation} \label{eq:needed_inequality}
	\Exp{f(\tilde{S} \cup \opt) \mid |S| < k}
	\geq
	(1 - \max_{e \in \groundset} \Pr[e \in \tilde{S}]) \cdot f(\opt)
	\enspace.
\end{equation}
It is claimed by~\cite{chekuri2015streaming} that this inequality follows from Lemma~\ref{lem:atMostP} (originally due to~\cite{DBLP:conf/soda/BuchbinderFNS14}). However, this lemma can yield only the inequalities
\[
	\Exp{f(\tilde{S} \cup \opt)}
	\geq
	(1 - \max_{e \in \groundset} \Pr[e \in \tilde{S}]) \cdot f(\opt)
\]
and
\[
	\Exp{f(\tilde{S} \cup \opt) \mid |S| < k}
	\geq
	(1 - \max_{e \in \groundset} \Pr[e \in \tilde{S} \mid |S| < k]) \cdot f(\opt)
	\enspace,
\]
which are similar to~\eqref{eq:needed_inequality}, but do not imply it.

\section{Inapproximability} \label{app:inapproximability}
In this appendix, we prove an inapproximability result for the problem of maximizing a non-negative submodular function subject to cardinality constraint in the data stream model. This result is given by the next theorem. The proof of the theorem is an adaptation of a proof given by Buchbinder et al.~\cite{BuchbinderFS19} for a similar result applying to an online variant of the same problem.

\begin{theorem}
For every constant $\varepsilon > 0$, no data stream algorithm for maximizing a non-negative submodular function subject to cardinality constraint is $(1/2 + \varepsilon)$-competitive, unless it uses $\Omega(\power{\groundset})$ memory.
\end{theorem}
\begin{proof}
Let $k \geq 1$ and $h \geq 1$ be two integers to be chosen later, and consider the non-negative submodular function $\Func{f}{2^\groundset}{\realnum^+}$, where $\groundset=\curly{u_i}_{i=1}^{k-1}\cup\curly{v_i}_{i=1}^{h}\cup\curly{w}$,
	defined as follows.
	\[
		f(S)=
		\begin{cases}
			\power{S} &\text{if }w\notin S \enspace,\\
			k+\power{S\cap\curly{u_i}_{i=1}^{k-1}} &\text{if }w\in S\enspace.
		\end{cases}
	\]
It is clear that $f$ is non-negative. One can also verify that the marginal value of each element in $\groundset$ is non-increasing, and hence, $f$ is submodular. 

Let $ALG$ be an arbitrary data stream algorithm for the problem of maximizing a non-negative submodular function subject to a cardinality constraint, and let us consider what happens when we give this algorithm the above function $f$ as input, the last element of $\groundset$ to arrive is the element $w$ and we ask the algorithm to pick a set of size at most $k$. One can observe that, before the arrival of $w$, $ALG$ has no way to distinguish between the other elements of $\groundset$. Thus, if we denote by $M$ the set of elements stored by $ALG$ immediately before the arrival of $w$ and assume that the elements of $\groundset \setminus \{w\}$ arrive at a random order, then every element of $\groundset \setminus \{w\}$ belongs to $M$ with the same probability of $\Exp{|M|} / |\groundset \setminus \{w\}|$. Hence, there must exist some arrival order for the elements of $\groundset \setminus \{w\}$ guaranteeing that
\[
	\Exp{|M \cap \curly{u_i}_{i=1}^{k-1}|}
	=
	\sum_{i = 1}^{k - 1} \Pr[u_i \in M]
	\leq
	\frac{k \cdot \Exp{|M|}}{|\groundset \setminus \{w\}|}
	\enspace.
\]

Note now that the above implies that the expected value of the output set produced by $ALG$ given the above arrival order is at most
\[
	k + \frac{k \cdot \Exp{|M|}}{|\groundset \setminus \{w\}|}
	\enspace.
\]
In contrast, the optimal solution is the set $\curly{u_i}_{i=1}^{k-1}\cup\curly{w}$, whose value is $2k-1$. Therefore, the competitive ratio of $ALG$ is at most
\[
	\frac{k + k \cdot \Exp{|M|} / |\groundset \setminus \{w\}|}{2k - 1}
	=
	\frac{1 + \Exp{|M|} / |\groundset \setminus \{w\}|}{2 - 1/k}
	\leq
	\frac{1}{2} + \frac{1}{k} + \frac{\Exp{|M|}}{|\groundset \setminus \{w\}|}
	\enspace.
\]
To prove the theorem we need to show that, when the memory used by $ALG$ is $o(|\groundset|)$, we can choose large enough values for $k$ and $h$ that will guarantee that the rightmost side of the last inequality is at most $1/2 + \varepsilon$. We do so by showing that the two terms $1/k$ and $\Exp{|M|} / |\groundset \setminus \{w\}|$ can both be upper bounded by $\varepsilon / 2$ when the integers $k$ and $h$ are large enough, respectively. For the term $1/k$ this is clearly the case when $k$ is larger than $2 / \varepsilon$. For the term $\Exp{|M|} / |\groundset \setminus \{w\}|$ this is true because increasing $h$ can make $\groundset$ as large as want, and thus, can make the ratio $\Exp{|M|} / |\groundset \setminus \{w\}|$ as small as necessary due to our assumption that the memory used by $ALG$ (which includes $M$) is $o(|\groundset|)$.
\end{proof}

\section{Multilinear extension based algorithm}
\label{app:continuous}

The properties of the multlinear extension based variant of our algorithm are summerized by the following theorem.
\newcommand{\TrmContinuous}[1][]{Assume there exists an $\alpha$-approximation offline algorithm $\offlinealg$ for maximizing a non-negative submodular function subject to cardinality constraint whose space complexity is nearly linear in the size of the ground set. Then, for every constant $\varepsilon \in (0,1]$, there exists an $(\frac{\alpha}{1 + \alpha} - \varepsilon)$-approximation semi-streaming algorithm for maximizing a non-negative submodular function subject to a cardinality constraint. The algorithm stores at most $O(k\varepsilon^{-2}\ifthenelse{\isempty{#1}}{\log \alpha^{-1}}{})$ elements.\ifthenelse{\isempty{#1}}{}{\footnote{Formally, the number of elements stored by the algorithm also depends on $\log \alpha^{-1}$. Since $\alpha$ is typically a positive constant, or at least lower bounded by a positive constant, we omit this dependence from the statement of the theorem.}}}
\begin{theorem} \label{trm:continuous}
\TrmContinuous[]
\end{theorem}

For ease of the reading, we present below only a simplified version of the multlinear extension based variant of our algorithm. This simplified version (given as Algorithm~\ref{alg:main_algorithm}) captures our main new ideas, but makes two simplifying assumptions that can be avoided using standard techniques.
\begin{itemize}
	\item The first assumption is that Algorithm~\ref{alg:main_algorithm} has access to an estimate $\tau$ of $f(\opt)$ obeying $(1 - \varepsilon/8) \cdot f(\opt) \leq \tau \leq f(\opt)$. Such an estimate can be produced using well-known techniques, such as a technique of~\cite{DBLP:conf/icml/0001MZLK19} used in Section~\ref{sec:full}, at the cost of increasing the space complexity of the algorithm only by a factor of $O(\varepsilon^{-1} \log \alpha^{-1})$.
	\item The second assumption is that Algorithm~\ref{alg:main_algorithm} has value oracle access to the multilinear extension $F$. If the time complexity of Algorithm~\ref{alg:main_algorithm} is not important, then this assumption is of no consequence since a value oracle query to $F$ can be emulated using an exponential number of value oracle queries to $f$. However, the assumption becomes problematic when we would like to keep the time complexity of the algorithm polynomial and we only have value oracle access to $f$, in which case this assumption can be dropped using standard sampling techniques (such as the one used in~\cite{DBLP:journals/siamcomp/CalinescuCPV11}). Interestingly, the rounding step of the algorithm and the sampling technique are the only parts of the extension based algorithm that employ randomness. Since the rounding can be made deterministic given either exponential time or value oracle access to $F$, we get the following observation.
	\begin{observation}
	If $\offlinealg$ is deterministic, then our multinear extension based algorithm is also \emph{deterministic} when it is allowed either exponential computation time or value oracle access to $F$.
	\end{observation}
\end{itemize}
A more detailed discussion of the techniques for removing the above assumptions can be found in an earlier version of this paper (available in~\cite{oldArxivVersion}) that had this variant of our algorithm as one of its main results.

Algorithm~\ref{alg:main_algorithm} has two constant parameters $p \in (0, 1)$ and $c > 0$ and maintains a fractional solution $x \in [0, 1]^\groundset$. This fractional solution starts empty, and the algorithm adds to it fractions of elements as they arrive. Specifically, when an element $e$ arrives, the algorithm considers its marginal contribution with respect to the current fractional solution $x$. If this marginal contribution exceeds the threshold of $c\tau / k$, then the algorithm tries to add to $x$ a $p$-fraction of $e$, but might end up adding a smaller fraction of $e$ if adding a full $p$-fraction of $e$ to $x$ will make $x$ an infeasible solution, i.e., make $\|x\|_1 > k$ (note that $\|x\|_1$ is the sum of the coordinates of $x$).

After viewing all of the elements, Algorithm~\ref{alg:main_algorithm} uses the fractional solution $x$ to generate two sets $S_1$ and $S_2$ that are feasible (integral) solutions. The set $S_1$ is generated by rounding the fractional solution $x$. As mentioned in Section~\ref{sec:simplified}, two rounding procedures, named Pipage Rounding and Swap Rounding, were suggested for this task in the literature~\cite{DBLP:journals/siamcomp/CalinescuCPV11,DBLP:conf/focs/ChekuriVZ10}. Both procedures run in polynomial time and guarantee that the output set $S_1$ of the rounding is always feasible, and that its expected value with respect to $f$ is at least the value $F(x)$ of the fractional solution $x$. The set $S_2$ is generated by applying $\offlinealg$ to the support of the vector $x$, which produces a feasible solution that (approximately) maximizes $f$ among all subsets of the support whose size is at most $k$. After computing the two feasible solutions $S_1$ and $S_2$, Algorithm~\ref{alg:main_algorithm} simply returns the better one of them.

\begin{algorithm2e}[ht]
\caption{\streamcont (simplified) $(p,c)$} \label{alg:main_algorithm}
	\DontPrintSemicolon
	Let $x\gets\onevect{\varnothing}$.\\
	\For {each arriving element $e$} {
		\lIf {$\partial_e F(x)\geq\frac{c\tau}{k}$} {
			$x\gets{x}+\min\curly{p,k-\onenorm{x}}\cdot\onevect{e}$.
		}
	}
	Round the vector $x$ to yield a feasible solution $S_1$ such that $\Exp{f(S_1)} \geq {F(x)}$.\\
	Find another feasible solution $S_2\subseteq\supp(x)$ by running $\offlinealg$ with $\supp(x)$ as the ground set.\\
	\Return the better solution among $S_1$ and $S_2$.
\end{algorithm2e}

Let us denote by $\hat{x}$ the final value of the fractional solution $x$ (i.e., its value when the stream ends). We begin the analysis of Algorithm~\ref{alg:main_algorithm} with the following useful observation. In the statement of observation, and in the rest of the section, we denote by $\supp(x)$ the support of vector $x$, i.e., the set $\curly{e\in\groundset\mid x_e>0}$.
\newcommand{\ObsXType}{If $\|\hat{x}\|_1 < k$, then $\hat{x}_e = p$ for every $e \in \supp(\hat{x})$. Otherwise, $\|\hat{x}\|_1 = k$, and the first part of the observation is still true for every element $e \in \supp(\hat{x})$ except for maybe a single element.}

\begin{observation} \label{obs:x_type}
\ObsXType
\end{observation}
\begin{proof}
For every element $e$ added to the support of $x$ by Algorithm~\ref{alg:main_algorithm}, the algorithm sets $x_e$ to $p$ unless this will make $\|x\|_1$ exceed $k$, in which case the algorithm set $x_e$ to be the value that will make $\|x\|_1$ equal to $k$. Thus, after a single coordinate of $x$ is set to a value other than $p$ (or the initial $0$), $\|x\|_1$ becomes $k$ and Algorithm~\ref{alg:main_algorithm} stops changing $x$.
\end{proof}

Using the last observation we can now bound the space complexity of Algorithm~\ref{alg:main_algorithm}, and show (in particular) that it is a semi-streaming algorithm for a constant $p$ when the space complexity of $\offlinealg$ is nearly linear.

\begin{observation} \label{obs:space_simple}
Algorithm~\ref{alg:main_algorithm} can be implemented so that it stores at most $O(k/p)$ elements.
\end{observation}
\begin{proof}
To calculate the sets $S_1$ and $S_2$, Algorithm~\ref{alg:main_algorithm} needs access only to the elements of $\groundset$ that appear in the support of $x$. Thus, the number of elements it needs to store is $O(|\supp(\hat{x})|) = O(k/p)$, where the equality follows from Observation~\ref{obs:x_type}.
\end{proof}

We now divert our attention to analyzing the approximation ratio of Algorithm~\ref{alg:main_algorithm}. The first step in this analysis is lower bounding the value of $F(\hat{x})$, which we do by considering two cases, one when $\onenorm{\hat{x}}=k$, and the other when $\onenorm{\hat{x}}<k$. The following lemma bounds the value of $F(\hat{x})$ in the first of these cases. Intuitively, this lemma holds since $\supp(\hat{x})$ contains many elements, and each one of these elements must have increased the value of $F(x)$ significantly when added (otherwise, Algorithm~\ref{alg:main_algorithm} would not have added this element to the support of $x$).
\newcommand{\LemFxExactlyK}{If $\onenorm{\hat{x}}=k$, then $F(\hat{x})\geq c\tau$.}
\begin{lemma}
\label{lem:Fx-exactlyK}
	\LemFxExactlyK
\end{lemma}
\begin{proof}
	Denote by $e_1,e_2,\dots,e_\ell$ the elements in the support of $\hat{x}$, in the order of their arrival. Using this notation, the value of $F(\hat{x})$ can be written as follows.
	\begin{align*}
		F(\hat{x})&=F(\onevect{\varnothing})+\sum_{i=1}^{\ell}{\Big(F\big(\hat{x}\wedge\onevect{\curly{e_1,e_2,\dotsc,e_i}}\big)-F\big(\hat{x}\wedge\onevect{\curly{e_1,e_2,\dotsc,e_{i-1}}}\big)\Big)}\\
		&=F(\onevect{\varnothing})+\sum_{i=1}^{\ell}{\Big(\hat{x}_{e_i}\cdot\partial_{e_i}F\big(\hat{x}\wedge\onevect{\curly{e_1,e_2,\dotsc,e_{i-1}}}\big)\Big)}\\
		&\geq{F(\onevect{\varnothing})}+\frac{c\tau}{k}\cdot\sum_{i=1}^{\ell}{\hat{x}_{e_i}}=F(\onevect{\varnothing})+\frac{c\tau}{k}\cdot\onenorm{\hat{x}}\geq c\tau\enspace,
	\end{align*}
where the second equality follows from the multilinearity of $F$, and the first inequality holds since Algorithm~\ref{alg:main_algorithm} selects an element $e_i$ only when $\partial_{e_i}F\big(\hat{x}\wedge\onevect{\curly{e_1,e_2,\dotsc,e_{i-1}}}\big)\geq\frac{c\tau}{k}$. The last inequality holds since $f$ (and thus, also $F$) is non-negative and $\|\hat{x}\|_1 = k$ by the assumption of the lemma.
\end{proof}

Consider now the case in which $\onenorm{\hat{x}}<k$. Recall that our objective is to lower bound $F(\hat{x})$ in this case as well. To do that, we need to a tool for upper bounding the possible increase in the value of $F(x)$ when some of the indices of $x$ are zeroed. The next lemma provides such an upper bound.

\newcommand{\CorDiscardElements}{Let $\Func{f}{2^{\groundset}}{\nonnegative}$ be a non-negative submodular function, let $p$ be a number in the range $\range{0}{1}$ and let $x,y$ be two vectors in $\cube{0}{1}{\groundset}$ such that
\begin{itemize}
	\item
	$\supp(x) \cap \supp(y) = \varnothing$,
	\item
	and	$y_e\leq p$ for every $e\in\groundset$.
\end{itemize}
Then, the multilinear extension $F$ of $f$ obeys $F(x+y)\geq(1-p)\cdot F(x)$.}
\begin{lemma}
\label{lem:discardElements}
	\CorDiscardElements
\end{lemma}
\begin{proof}
	Let us define the function $G_x(S)=\Exp{f(\RFunc(x)\cup S)}$. It is not difficult to verify that $G_x$ is non-negative and submodular, and that $G_x(\varnothing)=F(x)$. Additionally, since $\supp(x) \cap \supp(y) = \varnothing$, $\RFunc(x+y)$ has the same distribution as $\RFunc(x)\cup\RFunc(y)$, and therefore,
	\begin{align*}
		F(x+y)={}&\Exp{f\big(\RFunc(x + y))} = \Exp{f\big(\RFunc(x)\cup\RFunc(y)\big)}\\={}&\Exp{G_x(\RFunc(y))}\geq(1-p)\cdot{G_x(\varnothing)}=(1-p)\cdot{F(x)}\enspace,
	\end{align*}
	where the inequality follows from Lemma~\ref{lem:atMostP}.
\end{proof}

Using the last lemma, we now prove two lemmata proving upper and lower bounds the expression $F(\hat{x}+\onevect{\opt\setminus\supp(\hat{x})})$.

\begin{lemma}
\label{lem:belowK_lower}
	If $\onenorm{\hat{x}}<k$, then $F\big(\hat{x}+\onevect{\opt\setminus\supp(\hat{x})}\big)\geq
	(1-p)\cdot\big[p\cdot f(\opt)+(1-p)\cdot f\big(\opt\setminus\supp(\hat{x})\big)\big]$.
\end{lemma}
\begin{proof}
	Since $\onenorm{\hat{x}}<k$, Observation~\ref{obs:x_type} guarantees that $\hat{x}_e = p$ for every $e\in\supp(\hat{x})$. Thus $\hat{x}=p\cdot\onevect{\opt\cap\supp(\hat{x})}+p\cdot\onevect{\supp(\hat{x})\setminus \opt}$, and therefore,
	\begin{align*}
		F\big(\hat{x}+\onevect{\opt\setminus\supp(\hat{x})}\big)&=
		F\big(p\cdot\onevect{\opt\cap\supp(\hat{x})}
		+p\cdot\onevect{\supp(\hat{x})\setminus \opt}
		+\onevect{\opt\setminus\supp(\hat{x})}\big)\\
		&\geq (1-p)\cdot F\big(p\cdot\onevect{\opt\cap\supp(\hat{x})}+\onevect{\opt\setminus\supp(\hat{x})}\big)\\
		&\geq (1-p)\cdot \hat{f}\big(p\cdot\onevect{\opt\cap\supp(\hat{x})}+\onevect{\opt\setminus\supp(\hat{x})}\big)\\
		&=(1-p)\cdot\Big[p\cdot f(\opt)+(1-p)\cdot f\big(\opt\setminus\supp(\hat{x})\big)\Big]\enspace,
	\end{align*}
	where the first inequality follows from Lemma~\ref{lem:discardElements}, the second inequality holds
	since the Lov\'{a}sz extension lower bounds the multilinear extension, and the last equality follows from the definition of the Lov\'{a}sz extension.
\end{proof}

In the following lemma, and the rest of the section, we use the notation $b = k^{-1} \cdot |\opt \setminus \supp(\hat{x})|$. Intuitively, the lemma holds since the fact that the elements of $\opt \setminus \supp(\hat{x})$ were not added to the support of $x$ by Algorithm~\ref{alg:main_algorithm} implies that their marginal contribution is small.
\newcommand{\LemBelowKUpper}{If $\onenorm{\hat{x}}<k$, then $F\big(\hat{x}+\onevect{\opt\setminus\supp(\hat{x})}\big)\leq F(\hat{x})+bc\tau$.}
\begin{lemma}\label{lem:belowK_upper}
	\LemBelowKUpper
\end{lemma}
\begin{proof}
	The elements in $\opt\setminus\supp(\hat{x})$ were rejected by Algorithm~\ref{alg:main_algorithm}, which means that their marginal contribution with respect to the fractional solution $x$ at the time of their arrival was smaller than $c\tau/k$. Since the fractional solution $x$ only increases during the execution of the algorithm, the submodularity of $f$ guarantees that the same is true also with respect to $\hat{x}$. More formally, we get
\[
	\partial_e F(\hat{x})<\frac{c\tau}{k}
	\quad \forall\; e\in  \opt \setminus \supp(\hat{x})
	\enspace.
\]
Using the submodularity of $f$ again, this implies
	\[
		F\big(\hat{x}+\onevect{\opt\setminus\supp(\hat{x})}\big)
			\leq F(\hat{x})+\mspace{-40mu}\sum_{e\in \opt\setminus\supp(\hat{x})}
			{\mspace{-40mu}\partial_e F(\hat{x})}
			\leq F(\hat{x})+|\opt \setminus \supp(\hat{x})| \cdot \frac{c\tau}{k}
			= F(\hat{x})+bc\tau\enspace.		\qedhere
	\]
\end{proof}

Combining the last two lemmata immediately yields the promised lower bound on $F(\hat{x})$. To understand the second inequality in the following corollary, recall that $\tau \leq f(\opt)$.
\begin{corollary}
\label{cor:belowK}
	If $\onenorm{\hat{x}}<k$, then $F(\hat{x})\geq
	(1-p)\cdot\Big[p\cdot f(\opt)+(1-p)\cdot f\big(\opt\setminus\supp(\hat{x})\big)\Big]-bc\tau
	\geq
	[p(1 - p) - bc]\tau + (1 - p)^2 \cdot f\big(\opt\setminus\supp(\hat{x})\big)$.
\end{corollary}

Our next step is to get a lower bound on the expected value of $f(S_2)$. One easy way to get such a lower bound is to observe that $\opt \cap \supp(\hat{x})$ is a subset of the support of $\hat{x}$ of size at most $k$, and thus, is a feasible solution for $\offlinealg$ to return; which implies $\Exp{f(S_2)} \geq \alpha \cdot f(\opt \cap \supp(\hat{x}))$ since the algorithm $\offlinealg$ used to find $S_2$ is an $\alpha$-approximation algorithm. The following lemma proves a more involved lower bound by considering the vector $(b \hat{x}) \vee \onevect{\opt \cap \supp(\hat{x})}$ as a fractional feasible solution (using the rounding methods discussed above it, it can be converted into an integral feasible solution of at least the same value). The proof of the lemma lower bounds the value of the vector $(b \hat{x}) \vee \onevect{\opt \cap \supp(\hat{x})}$ using the concavity of the function $F((t \cdot \hat{x}) \vee \onevect{\opt \cap \supp(\hat{x})})$ as well as ideas used in the proofs of the previous claims.
\begin{lemma} \label{lem:belowK}
If $\onenorm{\hat{x}}<k$, then $\Exp{f(S_2)} \geq \alpha b(1 - p - cb)\tau + \alpha(1 - b) \cdot f(\opt \cap \supp(\hat{x}))$.
\end{lemma}
\begin{proof}
Consider the vector $(b \hat{x}) \vee \onevect{\opt \cap \supp(\hat{x})}$. Clearly,
\begin{align*}
	\onenorm{(b \hat{x}) \vee \onevect{\opt \cap \supp(\hat{x})}}
	\leq{} &
	b \cdot \onenorm{\hat{x}} + \onenorm{\onevect{\opt \cap \supp(\hat{x})}}\\
	\leq{} &
	|\opt \setminus \supp(\hat{x})| + |\opt \cap \supp(\hat{x})|
	=
	|\opt|
	\leq
	k
	\enspace,
\end{align*}
where the second inequality holds by the definition of $b$ since $\onenorm{\hat{x}} < k$.
Thus, due to the existence of the rounding methods discussed in the beginning of the section, there must exist a set $S$ of size at most $k$ obeying $f(S) \geq F((b \hat{x}) \vee \onevect{\opt \cap \supp(\hat{x})})$. Since $S_2$ is produced by $\offlinealg$, whose approximation ratio is $\alpha$, this implies $\Exp{f(S_2)} \geq \alpha \cdot F((b \hat{x}) \vee \onevect{\opt \cap \supp(\hat{x})})$. Thus, to prove the lemma it suffices to show that $F((b \hat{x}) \vee \onevect{\opt \cap \supp(\hat{x})})$ is always at least $b(1 - p - cb)\tau + (1 - b) \cdot f(\opt \cap \supp(\hat{x}))$.

The first step towards proving the last inequality is getting a lower bound on $F(\hat{x} \vee \onevect{\opt \cap \supp(\hat{x})})$. Recall that we already showed in the proof of Lemma~\ref{lem:belowK_upper} that
\[
	\partial_e F(\hat{x})<\frac{c\tau}{k}
	\quad \forall\; e\in  \opt \setminus \supp(\hat{x})
	\enspace.
\]
Thus, the submodularity of $f$ implies
\begin{align*}
	F(\hat{x} \vee \onevect{\opt}&)
	\leq
	F(\hat{x} \vee \onevect{\opt \cap \supp(\hat{x})}) + \sum_{e \in \opt \setminus \supp(\hat{x})} \mspace{-27mu} \partial_e F(\hat{x})\\
	\leq{} &
	F(\hat{x} \vee \onevect{\opt \cap \supp(\hat{x})}) + \frac{c\tau \cdot |\opt \setminus \supp(\hat{x})|}{k}
	=
	F(\hat{x} \vee \onevect{\opt \cap \supp(\hat{x})}) + cb\tau
	\enspace.
\end{align*}
Rearranging this inequality yields
\[
	F(\hat{x} \vee \onevect{\opt \cap \supp(\hat{x})})
	\geq
	F(\hat{x} \vee \onevect{\opt}) - cb\tau
	\geq
	(1 - p) \cdot f(\opt) - cb\tau
	\geq
	(1 - p - cb)\tau
	\enspace,
\]
where the second inequality holds by Lemma~\ref{lem:discardElements} since Observation~\ref{obs:x_type} guarantees that every coordinate of $\hat{x}$ is either $0$ or $p$. This gives us the promised lower bound on $F(\hat{x} \vee \onevect{\opt \cap \supp(\hat{x})})$.

We now note that the submodularity of $f$ implies that $F((t \cdot \hat{x}) \vee \onevect{\opt \cap \supp(\hat{x})})$ is a concave function of $t$ within the range $[0, 1]$. Since $b$ is inside this range,
\begin{align*}
	F((b \hat{x}) \vee \onevect{\opt \cap \supp(\hat{x})})
	\geq{} &
	b \cdot F(\hat{x} \vee \onevect{\opt \cap \supp(\hat{x})})
	+
	(1 - b) \cdot f(\opt \cap \supp(\hat{x}))\\
	\geq{} &
	b(1 - p - cb)\tau + (1 - b) \cdot f(\opt \cap \supp(\hat{x}))
	\enspace,
\end{align*}
which completes the proof of the lemma.
\end{proof}

Using the last two claims we can now obtain a lower bound on the value of the solution of Algorithm~\ref{alg:main_algorithm} in the case of $\onenorm{\hat{x}} < k$ which is a function of $\alpha$, $\tau$ and $p$ alone. We note that both the guarantees of Corollary~\ref{cor:belowK} and Lemma~\ref{lem:belowK} are lower bounds on the expected value of the output of the algorithm in this case since $\Exp{f(S_1)} \geq F(\hat{x})$. Thus, any convex combination of these guarantees is also such a lower bound, and the proof of the following corollary basically proves a lower bound for one such convex combination---for the specific value of $c$ stated in the corollary.
\newcommand{\CorBelowKConclusion}{If $\onenorm{\hat{x}}<k$ and $c$ is set to $\frac{\alpha(1 - p)}{\alpha + 1}$, then $\Exp{\max\{f(S_1), f(S_2)\}}\geq \frac{(1 - p)\alpha \tau}{\alpha + 1}$.}
\begin{corollary}
\label{cor:belowK_conclusion}
\CorBelowKConclusion
\end{corollary}
\begin{proof}
The corollary follows immediately from the non-negativity of $f$ when $p = 1$. Thus, we may assume $p < 1$ in the rest of the proof.

By the definition of $S_1$, $\Exp{f(S_1)} \geq F(\hat{x})$. Thus, by Corollary~\ref{cor:belowK} and Lemma~\ref{lem:belowK},
\begin{align*}
		\Exp{\max\curly{f(S_1),f(S_2)}}\mspace{-81mu}&\mspace{81mu}
		\geq
		\max\curly{\Exp{f(S_1)},\Exp{f(S_2)}}\\
		\geq{}&
		\max\{[p(1 - p) - bc]\tau + (1 - p)^2 \cdot f\big(\opt\setminus\supp(\hat{x})\big), \\&\alpha b(1 - p - cb)\tau + \alpha (1 - b) \cdot f(\opt \cap \supp(\hat{x}))\}\\
		\geq{} &
		\frac{\alpha(1 - b)}{\alpha(1 - b)+(1-p)^2}\cdot
		\bracks{[p(1 - p) - bc]\tau + (1 - p)^2 \cdot f\big(\opt\setminus\supp(\hat{x})\big)}\\
		&+\frac{(1-p)^2}{\alpha(1 - b)+(1-p)^2}\cdot \bracks{\alpha b(1 - p - cb)\tau + \alpha (1 - b) \cdot f(\opt \cap \supp(\hat{x}))}
		\enspace.
	\end{align*}
To keep the following calculations short, it will be useful to define $q = 1 - p$ and $d = 1 - b$. Using this notation and the fact that the submodularity and non-negativity of $f$ guarantee together $f\big(\opt\setminus\supp(\hat{x})\big) + f\big(\opt \cap \supp(\hat{x})\big) \geq f(\opt) \geq \tau$, the previous inequality implies
	\begin{align} \label{eq:lower_bound_with_b}
		\frac{\Exp{\max\curly{f(S_1),f(S_2)}}}{\alpha\tau}
		\geq\mspace{-120mu}{} &\mspace{120mu}
		\frac{(1 - b)[p(1 - p) - bc] + b(1 - p)^2(1 - p - bc) + (1 - b)(1 - p)^2}{\alpha(1 - b)+(1-p)^2} \nonumber\\
		={} &
		\frac{d[q(1 - q) - (1 - d)c] + q^2(1 - d)[q - (1 - d)c] + dq^2}{\alpha d+q^2} \nonumber\\
		={} &
		\frac{d[q - (1 - d)c] + q^2(1 - d)[q - (1 - d)c]}{\alpha d+q^2}
		=
		\frac{[d + q^2(1 - d)][q - (1 - d)c]}{\alpha d+q^2}\nonumber\\
		={} &
		\frac{q[d + q^2(1 - d)](d\alpha + 1)}{(\alpha + 1)(d\alpha+q^2)}
		=
		\frac{d^2\alpha + d\alpha q^2 - d^2\alpha q^2 + d + q^2 - dq^2}{d\alpha+q^2} \cdot \frac{q}{\alpha + 1}
		\enspace,
	\end{align}
	where the fourth equality holds by plugging in the value we assume for $c$. 
	
	The second fraction in the last expression is independent of the value of $d$, and the derivative of the first fraction in this expression as a function of $d$ is
	\begin{align*}
		\mspace{300mu}&\mspace{-300mu}
		\frac{(2d\alpha + \alpha q^2 - 2d\alpha q^2 + 1 - q^2)[d\alpha+q^2] - \alpha(d^2\alpha + d\alpha q^2 - d^2\alpha q^2 + d + q^2 - dq^2)}{[d\alpha+q^2]^2}\\
		={}&
		\frac{1 - q^2}{[d\alpha+q^2]^2} \cdot [q^2(1 - \alpha) + d\alpha(d\alpha + 2q^2)]
		\enspace,
	\end{align*}
	which is always non-negative since both $q$ and $\alpha$ are numbers between $0$ and $1$. Thus, we get that the minimal value of the expression~\eqref{eq:lower_bound_with_b} is obtained for $d = 0$ for any choice of $q$ and $\alpha$. Plugging this value into $d$ yields
	\[
		\Exp{\max\curly{f(S_1),f(S_2)}}
		\geq
		\frac{q\alpha \tau}{\alpha + 1}
		=
		\frac{(1 - p)\alpha \tau}{\alpha + 1}
		\enspace.
		\qedhere
	\]
\end{proof}

Note that Lemma~\ref{lem:Fx-exactlyK} and Corollary~\ref{cor:belowK_conclusion} both prove the same lower bound on the expectation $\Exp{\max\{f(S_1), f(S_2)\}}$ when $c$ is set to the value it is set to in Corollary~\ref{cor:belowK_conclusion} (because $\mathbb{E}[\max\{f(S_1),\allowbreak f(S_2)\}] \geq \Exp{f(S_1)} \geq F(\hat{x})$). Thus, we can summarize the results we have proved so far using the following proposition.
\begin{proposition}
\label{prp:results}
Algorithm~\ref{alg:main_algorithm} is a semi-streaming algorithm storing $O\left(k/p\right)$ elements. Moreover, for the value of the parameter $c$ given in Corollary~\ref{cor:belowK_conclusion}, the output set produced by this algorithm has an expected value of at least $\frac{\alpha\tau(1 - p)}{\alpha + 1}$.
\end{proposition}

Using the last proposition, we can now prove the following theorem. As discussed at the beginning of the section, the assumption that $\tau$ is known can be dropped at the cost of a slight increase in the number of of elements stored by the algorithm, which yields Theorem~\ref{trm:continuous}.
\begin{theorem}
\label{trm:non-polynomial-knowing-tau}
	For every constant $\varepsilon \in (0, 1]$, there exists a semi-streaming algorithm that assumes access to an estimate $\tau$ of $f(\opt)$ obeying $(1-\varepsilon/8) \cdot f(\opt)\leq\tau\leq{f(\opt)}$ and provides $(\frac{\alpha}{1 + \alpha} - \varepsilon)$-approximation for the problem of maximizing a non-negative submodular function subject to cardinality constraint. This algorithm stores at most $O(k\varepsilon^{-1})$ elements.
\end{theorem}
\begin{proof}
Consider the algorithm obtained from Algorithm~\ref{alg:main_algorithm} by setting $p = \varepsilon/2$ and $c$ as is set in Corollary~\ref{cor:belowK_conclusion}. By Proposition~\ref{prp:results}, this algorithm stores only $O(k/p) = O(k\varepsilon^{-1})$ elements, and the expected value of its output set is at least
\[
	\frac{\alpha\tau(1 - p)}{\alpha + 1}
	\geq
	\frac{\alpha(1 - \varepsilon/8)(1 - \varepsilon/2)}{\alpha + 1} \cdot f(\opt)
	\geq
	\frac{\alpha(1 - \varepsilon)}{\alpha + 1} \cdot f(\opt)
	\geq
	\left(\frac{\alpha}{\alpha + 1} - \varepsilon\right)\cdot f(\opt)
	\enspace,
\]
where the first inequality holds since $\tau$ obeys, by assumption, $\tau \geq (1 - \varepsilon/8) \cdot f(\opt)$.
\end{proof}

\section{Algorithm with true randomization} \label{app:randomized}

The variant of our algorithm that involves true randomization is shown in Algorithm~\ref{alg:randomized}. For simplicity, we
describe the algorithm assuming the knowledge of an estimate of the
value of the optimal solution, $f(\opt)$. To remove this assumption,
we use the standard technique introduced by \cite{badanidiyuru2014}.
The basic idea is to use the maximum singleton value $v=\max_{e}f(\{e\})$
as a $k$-approximation of $f(\opt)$. Given this approximation, one
can guess a $1+\varepsilon$ approximation of $f(\opt)$ from a set of
$O(\log(k/\alpha)/\varepsilon)$ values ranging from $v$ to $kv/\alpha$
(recall that $\alpha$ is the approximation guarantee of the offline algorithm
$\offlinealg$ that we use in the post-processing step). The final
streaming algorithm is simply $O(\log(k/\alpha)/\varepsilon)$ copies
of the basic algorithm running in parallel with different guesses.
As new elements appear in the stream, the value $v=\max_{e}f(\{e\})$
also increases over time and thus, existing copies of the basic algorithm
with small guesses are dropped and new copies with higher guesses
are added. An important observation is that when we introduce a new
copy with a large guess, starting it from mid-stream has exactly the
same outcome as if we started it from the beginning of the stream:
all previous elements have marginal gain much smaller than the guess
and smaller than the threshold so they would have been rejected anyway.
We refer to \cite{badanidiyuru2014} for the full details.
\begin{theorem}
There is a streaming algorithm $\stream$ for non-negative, non-monotone submodular
maximization with the following properties ($\varepsilon>0$ is any desired
accuracy and it is given as input to the algorithm):
\begin{itemize}
\item The algorithm makes a single pass over the stream.
\item The algorithm uses $O\left(\frac{k\log(k/\alpha)\log(1/\varepsilon)}{\varepsilon^{3}}\right)$
space.
\item The update time per item is $O\left(\frac{\log(k/\alpha)\log(1/\varepsilon)}{\varepsilon^{2}}\right)$
marginal gain computations.
\end{itemize}
At the end of the stream, we post-process the output of $\stream$
using any offline algorithm $\offlinealg$ for submodular maximization.
The resulting solution is a $\frac{\alpha}{1+\alpha}-\varepsilon$ approximation,
where $\alpha$ is the approximation of $\offlinealg$.
\end{theorem}

\begin{algorithm2e}[ht]
\caption{Streaming algorithm for $\max_{|S|\leq k} f(S)$. $\offlinepp$ uses any offline algorithm $\offlinealg$ with approximation $\alpha$. Lines shown in \textcolor{blue}{blue} are comments. The algorithm does \textbf{not} store the sets $V_{i,j}$, they are defined for analysis purposes only.} \label{alg:randomized}
\textcolor{blue}{$f\colon2^{V}\to\mathbb{R}_{\geq 0}$: submodular and non-negative}\\
\textcolor{blue}{$k$: cardinality constraint}\\
\textcolor{blue}{$\varepsilon\in(0,1]$: accuracy parameter}\\
\textcolor{blue}{$\threshold$: threshold}\\

$\underline{\stream(f,k,\varepsilon,\threshold)}$\\
$r\gets\Theta(\ln(1/\varepsilon)/\varepsilon)$\\
$m\gets1/\varepsilon$\\
$S_{i,j}\gets\varnothing$ for all $i\in[r],j\in[m]$\\
\textcolor{blue}{$V_{i,j}\gets\varnothing$ for all $i\in[r],j\in[m]$ //}
\textbf{\textcolor{blue}{{}not stored}}\textcolor{blue}{, defined for analysis purposes only}\\
\For{each arriving element $e$}{
  \For{$i=1$ to $r$}{
    choose an index $j\in[m]$ uniformly and independently at random\\
    \textcolor{blue}{$V_{i,j}\gets V_{i,j}\cup\{e\}$ //}\textbf{\textcolor{blue}{{}
not stored}}\textcolor{blue}{, defined for analysis purposes only}\\
    \If{$f(S_{i,j}\cup\{e\})-f(S_{i,j})\geq\threshold$ and $\left|S_{i,j}\right|<k$} {
      $S_{i,j}\gets S_{i,j}\cup\{e$\}
    }
  }
}
\Return $\left\{ S_{i,j}\colon i\in[r],j\in[m]\right\}$

\medskip{}

$\underline{\offlinepp(f,k,\varepsilon)}$\\
\textcolor{blue}{Assumes an estimate for $f(\opt)$, see text on how to remove this assumption}\\
\textcolor{blue}{Uses any offline algorithm $\offlinealg$ with approximation $\alpha$}\\
$\threshold\gets\frac{\alpha}{1+\alpha}\cdot\frac{1}{k}\cdot f(\opt)$ \textcolor{blue}{// threshold}\\
$\left\{ S_{i,j}\right\} \gets\stream(f,k,\varepsilon,\threshold)$\\
\If{$\left|S_{i,j}\right|=k$ for some $i$ and $j$} {
  \Return $S_{i,j}$
} \Else {
  $U\gets\bigcup_{i,j}S_{i,j}$\\
  $T\gets\offlinealg(f,k,U)$\\
  \Return $\arg\max\left\{ f(S_{1,1}),f(T)\right\}$
}
\end{algorithm2e}

\begin{algorithm2e}
\caption{Single threshold Greedy algorithm. The algorithm processes the elements
in the order in which they arrive in the stream, and it uses the same
threshold $\threshold$ as $\protect\stream$.} \label{alg:stgreedy}

$\underline{\stg(f,N,k,\threshold)}$:\\
$S\gets\varnothing$\\
\For{each $e\in N$ in the stream order}{
  \If{$f(S\cup\{e\})-f(S)\geq\threshold$ and $|S|<k$}{
    $S\gets S\cup\{e\}$
  }
}
\Return $S$
\end{algorithm2e}

In the remainder of this section, we analyze Algorithm~\ref{alg:randomized} and show that it achieves
a $\frac{\alpha}{1+\alpha}-\varepsilon$ approximation, where $\alpha$
is the approximation guarantee of the offline algorithm $\offlinealg$. 

We divide the analysis into two cases, depending on the probability
of the event that a set $S_{i,1}$ (for some $i\in[r]$) constructed
by $\stream$ has size $k$. For every $i\in[r]$, let $\cf_{i}$
be the event that $\left|S_{i,1}\right|=k$. Since each of the $r$
repetitions (iterations of the for loop of $\stream$) use independent
randomness to partition $V$, the events $\cf_{1},\dots,\cf_{r}$
are independent. Additionally, the events $\cf_{1},\dots,\cf_{r}$
have the same probability. We divide the analysis into two cases,
depending on whether $\pr{\cf_{1}}\geq\varepsilon$ or $\pr{\cf_{1}}<\varepsilon$.
In the first case, since we are repeating $r=\Theta(\ln(1/\varepsilon)/\varepsilon)$
times, the probability that there is a set $S_{i,j}$ of size $k$
is at least $1-\varepsilon$, and we obtain the desired approximation
since $f(S_{i,j})\geq\threshold\left|S_{i,j}\right|=\threshold k=\frac{\alpha}{1+\alpha}f(\opt)$.
In the second case, we have $\pr{\,\overline{\cf_{1}}\,}\geq1-\varepsilon$
and we argue that $\bigcup_{i,j}S_{i,j}$ contains a good solution.
We now give the formal argument for each of the cases.

\subsection{The case \texorpdfstring{$\pr{\protect\cf_{1}}\protect\geq\varepsilon$}{Pr[F1] > epsilon}}
As noted earlier, the events $\cf_{1},\dots,\cf_{r}$ are independent and
have the same probability. Thus,
\[
\pr{\,\overline{\cf_{1}\cup\dots\cup\cf_{r}}\,}\leq(1-\varepsilon)^{r}\leq\exp(-\varepsilon r)\leq\varepsilon
\]
since $r=\Theta(\ln(1/\varepsilon)/\varepsilon)$. Thus $\pr{\cf_{1}\cup\dots\cup\cf_{r}}\geq1-\varepsilon$. 

Conditioned on the event $\cf_{1}\cup\dots\cup\cf_{r}$, we obtain
the desired approximation due to the following lemma. The lemma follows
 from the fact that the marginal gain of each selected element is at least $\threshold$.
\begin{lemma}
\label{lem:stg}We have $f\left(S_{i,j}\right)\geq\threshold\left|S_{i,j}\right|$
for all $i\in[r]$, $j\in[m]$.
\end{lemma}
\begin{proof}
To simplify notation, let $S=S_{i,j}$.
Let $e_{1},\dots,e_{|S|}$ be the elements of $S$ in the order in
which they were added to $S$. Let $S^{(t)}=\{e_{1},\dots,e_{t}\}$.
We have $f(S^{(t)})-f(S^{(t-1)})=f(S^{(t-1)}\cup\{e_{t}\})-f(S^{(t-1)})\geq\threshold$
and thus
\begin{align*}
f(S)-f(\emptyset) & =\sum_{t=1}^{|S|}\left(f(S^{(t)})-f(S^{(t-1)})\right)\geq\threshold|S|
\end{align*}
\end{proof}

We can combine the two facts and obtain the desired approximation
as follows. Let $\cs$ be the random variable equal to the solution
returned by $\offlinepp$. We have
\begin{align*}
\Exp{f(\cs)} &\geq\Exp{f(\cs)\vert\cf_{1}\cup\dots\cup\cf_{r}}\pr{\cf_{1}\cup\dots\cup\cf_{r}}
 \geq(1-\varepsilon)\threshold k
 =(1-\varepsilon)\frac{\alpha}{1+\alpha}f(\opt)
\enspace.
\end{align*}

\subsection{The case \texorpdfstring{$\pr{\protect\cf_{1}}<\varepsilon$}{Pr[F1] < epsilon}}
In this case, we show that the solution $\arg\max\left\{ f(T),f(S_{1,1})\right\} $,
which is returned on the last line of $\offlinepp$, has good value in expectation.
Our analysis borrows ideas and techniques from the work of Barbosa
\emph{et al.} \cite{Barbosa2016}: the probabilities $p_{e}$ defined
below are analogous to the probabilities used in that work; the division
of $\opt$ into two sets based on these probabilities is analogous
to the division employed in Section 7.3 in that work; Lemma \ref{lem:stg-consistent}
shows a consistency property for the single threshold greedy algorithm
that is analogous to the consistency property shown for the standard
greedy algorithm and other algorithms by Barbosa \emph{et al}. We
emphasize that Barbosa \emph{et al}. use these concepts in a different
context (specifically, \emph{monotone} maximization in the \emph{distributed}
setting). When applied to our context---\emph{non-monotone} maximization
in the \emph{streaming} setting---the framework of Barbosa \emph{et
al.} requires $\Omega(\sqrt{nk})$ memory if used with a single pass
(alternatively, they use $\Omega(\min\{k,1/\varepsilon\})$ passes) and
achieves worse approximation guarantees.

\textbf{Notation and definitions.} For analysis purposes only, we
make use of the Lovasz extension $\hat{f}$. We fix an optimal solution
$\opt\in\arg\max\{f(A)\colon A\subseteq V,|A|\leq k\}$. Let $\mathcal{V}(1/m)$
be the distribution of $1/m$-samples of $V$, where a $1/m$-sample
of $V$ includes each element of $V$ independently at random with
probability $1/m$. Note that $V_{i,j}\sim\mathcal{V}(1/m)$ for every
$i\in[r]$, $j\in[m]$ (see $\stream$). Additionally, for each $i\in[r]$,
$V_{i,1},\dots,V_{i,m}$ is a partition of $V$ into $1/m$-samples. 

For a subset $N\subseteq V$, we let $\stg(N)$ be the output of the
single threshold greedy algorithm when run as follows (see also Algorithm~\ref{alg:stgreedy} for a formal description of the algorithm): the algorithm processes
the elements of $N$ \emph{in the order in which they arrive in the
stream} and it uses the same threshold $\threshold$ as $\stream$; starting
with the empty solution and continuing until the size constraint of
$k$ is reached, the algorithm adds an element to the current solution
if its marginal gain is above the threshold. Note that $S_{i,j}=\stg(V_{i,j})$
for all $i\in[r],j\in[m]$. For analysis purposes only, we also consider
$\stg(N)$ for sets $N$ that do not correspond to any set $V_{i,j}$.

For each $e\in V$, we define
\[
p_{e}=\begin{cases}
\Pr_{X\sim\mathcal{V}(1/m)}\left[e\in\stg(X\cup\{e\})\right] & \text{if }e\in\opt\\
0 & \text{otherwise} \enspace.
\end{cases}
\]
We partition $\opt$ into two sets:
\begin{align*}
O_{1} & =\left\{ e\in\opt\colon p_{e}\geq\varepsilon\right\} \\
O_{2} & =\opt\setminus O_{1} \enspace.
\end{align*}
We also define the following subset of $O_{2}$:
\[
O'_{2}=\left\{ e\in O_{2}\colon e\notin\stg\left(V_{1,1}\cup\{e\}\right)\right\}. 
\]
Note that $(O_{1},O_{2})$ is a deterministic partition of $\opt$,
whereas $O'_{2}$ is a random subset of $O_{2}$. The role of the sets $O_{1},O_{2},O'_{2}$ will become clearer in
the analysis. The intuition is that, using the repetition, we can
ensure that each element of $O_{1}$ ends up in the collected set
$U=\bigcup_{i,j}S_{i,j}$ with good probability: each iteration $i\in[r]$
ensures that an element $e\in O_{1}$ is in $S_{i,1}\cup\dots\cup S_{i,m}$
with probability $p_{e}\geq\varepsilon$ and, since we repeat $r=\Theta(\ln(1/\varepsilon)/\varepsilon)$
times, we will ensure that $\Exp{\one_{O_{1}\cap U}}\geq(1-\varepsilon)\one_{O_{1}}$.
We also have that $\Exp{\one_{O'_{2}}}\geq(1-\varepsilon)\one_{O_{2}}$:
an element $e\in O_{2}\setminus O'_{2}$ ends up being picked by $\stg$
when run on input $V_{1,1}\cup\{e\}$, which is a low probability
event for the elements in $O_{2}$; more precisely, the probability
of this event is equal to $p_{e}$ (since $V_{1,1}\sim\mathcal{V}(1/m)$)
and $p_{e}\leq\varepsilon$ (since $e\in O_{2}$). Thus $\Exp{\one_{(O_{1}\cap U)\cup O'_{2}}}\geq(1-\varepsilon)\one_{\opt}$
, which implies that the expected value of $(O_{1}\cap U)\cup O'_{2}$
is at least $(1-\varepsilon)f(\opt)$. However, whereas $O_{1}\cap U$
is available in the post-processing phase, elements of $O'_{2}$ may
not be available and they may account for most of the value of $O_{2}$.
The key insight is to show that $S_{1,1}$ makes up for the lost value
from these elements.

We start the analysis with two helper lemmas, which
follow from standard arguments that have been used in previous works.
The first of these lemmas follows from an argument based on the
Lovasz extension and its properties.

\newcommand{\LemSampleUnionOpt}{Let $0\leq u\leq v\leq1$. Let $S\subseteq V\setminus\opt$
and $O\subseteq\opt$ be random sets such that $\Exp{\one_{S}}\leq u\one_{V\setminus\opt}$
and $\Exp{\one_{O}}\geq v\one_{\opt}$. Then $\Exp{f(S\cup O)}\geq(v-u)f(\opt)$.}
\begin{lemma}
\label{lem:sample-union-opt}
\LemSampleUnionOpt
\end{lemma}
\begin{proof}
Let $\hat{f}$ be the Lovasz
extension of $f$. Using the fact that $\hat{f}$ is an extension and
it is convex, we obtain
\begin{align*}
\Exp{f(S\cup O)} & =\Exp{\hat{f}\left(\one_{S\cup O}\right)}\geq \hat{f}\left(\Exp{\one_{S\cup O}}\right)=\hat{f}\left(\Exp{\one_{S}}+\Exp{\one_{O}}\right)
\end{align*}
Let $\x:=\Exp{\one_{S}}+\Exp{\one_{O}}$. Note
that
\[
x_{e}=\begin{cases}
\pr{e\in S}\le u & \text{if }e\in V\setminus\opt\\
\pr{e\in O}\geq v & \text{if }e\in\opt \enspace.
\end{cases}
\]
Thus, we have
\[
\hat{f}(\x)=\int_{0}^{1}f\left(\left\{ e\in V\colon x_{e}\geq\theta\right\} \right)d\theta\geq\int_{u}^{v}f\left(\left\{ e\in V\colon x_{e}\geq\theta\right\} \right)d\theta=(v-u)f(\opt) \enspace.
\]
The first equality is the definition of $\hat{f}$. The inequality is
by the non-negativity of $f$. The second equality is due to the fact
that, for $u<\theta\leq v$, we have $\left\{ e\in V\colon x_{e}\geq\theta\right\} =\opt$.
\end{proof}

The following lemma establishes a consistency property for the $\stg$
algorithm, analogous to the consistency property shown and used by
Barbosa \emph{et al. }for algorithms such as the standard Greedy algorithm.
The proof is also very similar to the proof shown by Barbosa \emph{et
al.}
\newcommand{\LemStgConsistent}{Conditioned on the event $\left|S_{1,1}\right|<k$,
$\stg\left(V_{1,1}\cup O'_{2}\right)=\stg\left(V_{1,1}\right)=S_{1,1}$.}
\begin{lemma}
\label{lem:stg-consistent}
\LemStgConsistent
\end{lemma}
\begin{proof}
To simplify notation, we let
$V_{1}=V_{1,1}$ and $S_{1}=S_{1,1}$. Let $X=\stg(V_{1}\cup O'_{2})$.
Suppose for contradiction that $S_{1}\neq X$. Let $e_{1},e_{2},\dots,e_{|V_{1}\cup O'_{2}|}$
be the elements of $V_{1}\cup O'_{2}$ in the order in which they
arrived in the stream. Let $i$ be the smallest index such that $\stg(\{e_{1},\dots,e_{i}\})\neq\stg(\{e_{1},\dots,e_{i}\}\cap V_{1})$.
By the choice of $i$, we have 
\[\stg(\{e_{1},\dots,e_{i-1}\})=\stg(\{e_{1},\dots,e_{i-1}\}\cap V_{1}):=A \enspace.\]
Note that $|A|<k$, since $A\subseteq S_{1}$ and $|S_{1}|<k$ by
assumption. Since $\stg(\{e_{1},\dots,e_{i}\})\neq\stg(\{e_{1},\dots,e_{i}\}\cap V_{1})$,
we must have $e_{i}\notin V_{1}$ (and thus $e_{i}\in O'_{2}\setminus V_{1}$)
and $f(A\cup\{e_{i}\})-f(A)\geq\threshold$. The latter implies that $e_{i}\in\stg(V_{1}\cup\{e_{i}\})$:
after processing all of the elements of $V_{1}$ that arrived before
$e_{i}$, the partial greedy solution is $A$; when $e_{i}$ arrives,
it is added to the solution since $|A|<k$ and $f(A\cup\{e_{i}\})-f(A)\geq\threshold$.
But then $e_{i}\notin O'_{2}$, which is a contradiction.
\end{proof}

We now proceed with the main analysis. Recall that $\offlinepp$ runs
$\offlinealg$ on $U$ to obtain a solution $T$, and
returns the better of the two solutions $S_{1,1}$ and $T$. In the following
lemma, we show that the value of this solution is proportional to
$f(S_{1,1}\cup(O_{1}\cap U))$. Note that $S_{1,1}\cup(O_{1}\cap U)$
may not be feasible, since we could have $\left|S_{1,1}\right|>\left|O_{2}\right|$,
and hence the scaling based on $\frac{\left|O_{2}\right|}{k}$.
\begin{lemma}
\label{lem:T-val}
We have 
$\max\left\{ f(S_{1,1}),f(T)\right\} \geq\frac{\alpha}{1+\alpha\left(1-\frac{\left|O_{2}\right|}{k}\right)}f(S_{1,1}\cup(O_{1}\cap U))$.
\end{lemma}

\begin{proof}
To simplify notation, we let $S_{1}=S_{1,1}$. Let $b=\left|O_{2}\right|$.
First, we analyze $f(T)$. Let $X\subseteq S_{1}$ be a random subset
of $S_{1}$ such that $|X|\leq b$ and $\Exp{\one_{X}}=\frac{b}{k}\one_{S_{1}}$.
We can select such a subset as follows: we first choose a permutation
of $S_{1}$ uniformly at random, and let $\tilde{X}$ be the first
$s:=\min\left\{ b,\left|S_{1}\right|\right\} $ elements in the permutation.
For each element of $\tilde{X}$, we add it to $X$ with probability
$p:=|S_{1}|b / (sk)$.

Since $X\cup((O_{1}\cap U)\setminus S_{1})$ is a feasible solution
contained in $U$ and $\offlinealg$ achieves an $\alpha$-approximation,
we have
\[ f(T)\geq\alpha f(X\cup((O_{1}\cap U)\setminus S_{1})) \enspace.\]
By taking expectation over $X$ only (more precisely, the random sampling
that we used to select $X$) and using that $\hat{f}$ is a convex extension, we obtain:
\begin{align*}
f(T)
&\geq \alpha \Ex_{X}\left[f(X\cup((O_{1}\cap U)\setminus S_{1}))\right]
= \alpha \Ex_{X}\left[\hat{f}\left(\one_{X\cup((O_{1}\cap U)\setminus S_{1})}\right)\right]\\
&\geq \alpha \hat{f}\left(\Ex_{X}\left[\one_{X\cup((O_{1}\cap U)\setminus S_{1})} \right] \right)
= \alpha \hat{f}\left(\frac{b}{k}\one_{S_{1}} + \one_{(O_{1}\cap U)\setminus S_{1}}\right) \enspace.
\end{align*}
Next, we lower bound $\max\left\{ f(S_{1}),f(T)\right\} $ using a
convex combination $(1-\theta) f(S_{1}) + \theta f(T)$ with coefficient
$\theta =1/(1+\alpha\left(1-\frac{b}{k}\right))$. Note that $1-\theta=\theta\alpha\left(1-\frac{b}{k}\right)$. 
By taking this convex combination, using the previous inequality lower
bounding $f(T)$, and the convexity and restricted scale invariance of $\hat{f}$, we obtain:
\begin{align*}
&\max\left\{ f(S_{1}),f(T)\right\}
\geq (1-\theta)f(S_{1}) + \theta f(T)
=\theta\alpha \left(1-\frac{b}{k}\right) f(S_{1}) + \theta f(T)\\
&\geq \theta\alpha\left(1-\frac{b}{k}\right)\hat{f}\left(\one_{S_{1}}\right)+\theta\alpha\hat{f}\left(\frac{b}{k}\one_{S_{1}}+\one_{(O_{1}\cap U)\setminus S_{1}}\right)\\
&= \theta\alpha \left(2-\frac{b}{k}\right)\left(\frac{1-\frac{b}{k}}{2-\frac{b}{k}} \hat{f}\left(\one_{S_{1}}\right)+\frac{1}{2-\frac{b}{k}}\hat{f}\left(\frac{b}{k}\one_{S_{1}}+\one_{(O_{1}\cap U)\setminus S_{1}}\right)\right)\\
&\geq \theta\alpha \left(2-\frac{b}{k}\right) \hat{f}\left(\frac{1-\frac{b}{k}}{2-\frac{b}{k}}\one_{S_1}+\frac{1}{2-\frac{b}{k}}\left(\frac{b}{k}\one_{S_1}+\one_{(O_1\cap U)\setminus S_1}\right)\right)\\
&= \theta\alpha \left(2-\frac{b}{k}\right) \hat{f}\left(\frac{1}{2-\frac{b}{k}}\one_{S_1\cup(O_1\cap U)}\right)
\geq \frac{\alpha}{1+\alpha\left(1-\frac{b}{k}\right)}f(S_{1}\cup(O_{1}\cap U)) \enspace.
\qedhere
\end{align*}
\end{proof}

Next, we analyze the expected value of $f(S_{1,1}\cup(O_{1}\cap U))$.
We do so in two steps: first we analyze the marginal gain of $O'_{2}$
on top of $S_{1,1}$ and show that it is suitably small, and then
we analyze $f(S_{1,1}\cup(O_{1}\cap U)\cup O'_{2})$ and show that
its expected value is proportional to $f(\opt)$. We use the notation
$f(A\mid B)$ to denote the marginal gain of set $A$ on top of set $B$,
i.e., $f(A\mid B)=f(A\cup B)-f(B)$.

\newcommand{\LemOPrimeTwoGain}{We have $\Exp{f\left(O'_{2}\mid S_{1,1}\right)}\leq\threshold b+\varepsilon f(\opt)$.}
\begin{lemma}
\label{lem:O'2-gain}
\LemOPrimeTwoGain
\end{lemma}
\begin{proof}
As before, to simplify notation, we let $S_{1}=S_{1,1}$ and $V_{1}=V_{1,1}$.
We break down the expectation using the law of total expectation as
follows:
\begin{align*}
&\Exp{f\left(O'_{2}\mid S_{1}\right)}\\
& =\Exp{f\left(O'_{2}\mid S_{1}\right)\vert\left|S_{1}\right|<k}\cdot\underbrace{\pr{\left|S_{1}\right|<k}}_{\leq1}+\underbrace{\Exp{f\left(O'_{2}\mid S_{1}\right)\mid\left|S_{1}\right|=k}}_{\leq f(\opt)}\cdot\underbrace{\pr{\left|S_{1}\right|=k}}_{\leq\varepsilon}\\
 & \leq\Exp{f\left(O'_{2}\mid S_{1}\right)\mid\left|S_{1}\right|<k}+\varepsilon f(\opt) \enspace.
\end{align*}
Above, we have used that $f(O'_{2}\mid S_{1})\leq f(O'_{2})\leq f(\opt)$,
where the first inequality follows by submodularity. We have also
used that $\pr{\left|S_{1}\right|=k}=\pr{\cf_{1}}\leq\varepsilon$.
Thus, it only remains to show that $\Exp{f\left(O'_{2}\mid S_{1}\right)\vert\left|S_{1}\right|<k}\leq\threshold b$.

We condition on the event $\left|S_{1}\right|<k$ for the remainder
of the proof. By Lemma \ref{lem:stg-consistent}, we have $\stg(V_{1}\cup O'_{2})=S_{1}$.
Since $\left|S_{1}\right|<k$, each element of $O'_{2}\setminus S_{1}$
was rejected because its marginal gain was below the threshold when
it arrived in the stream. This, together with submodularity, implies
that
\[
f\left(O'_{2}\mid S_{1}\right)\leq\threshold\left|O'_{2}\right|\leq\threshold b \enspace.
\qedhere
\]
\end{proof}

\newcommand{\LemWithOPrimeTwo}{We have $\Exp{f(S_{1,1}\cup(O_{1}\cap U)\cup O'_{2})}\geq(1-2\varepsilon)f(\opt)$.}
\begin{lemma}
\label{lem:with-O'2}
\LemWithOPrimeTwo
\end{lemma}
\begin{proof}
We apply Lemma~\ref{lem:sample-union-opt} to the following sets:
\begin{align*}
S &= S_{1,1} \setminus \opt\\
O &= (S_{1,1}\cap \opt) \cup (O_{1}\cap U) \cup O'_{2} \enspace.
\end{align*}
We show below that $\Exp{\one_{S}} \leq \varepsilon \one_{V \setminus \opt}$ and $\Exp{\one_O} \geq (1-\varepsilon) \one_{\opt}$. Assuming these bounds, we can take $u=\varepsilon$ and $v=1-\varepsilon$ in Lemma~\ref{lem:sample-union-opt}, which gives the desired result.

Since $S\subseteq S_{1,1}\subseteq V_{1,1}$ and $V_{1,1}$ is a $(1/m)$-sample
of $V$, we have $\Exp{\one_{S}}\leq\frac{1}{m}\one_{V\setminus\opt}=\varepsilon\one_{V\setminus\opt}$.
Thus it only remains to show that, for each $e\in\opt$, we have $\pr{e\in O} \geq 1-\varepsilon$. Since $(O_{1}\cap U) \cup O'_{2} \subseteq O$, it suffices to show that $\pr{e \in (O_{1}\cap U) \cup O'_{2}} \geq 1-\varepsilon$, or equivalently that $\pr{e \in (O_1 \setminus U) \cup (O_2 \setminus O'_2)} \leq \varepsilon$.

Recall that $(O_{1},O_{2})$ is a deterministic partition of $\opt$.
Thus $e$ belongs to exactly one of $O_{1}$ and $O_{2}$ and we consider
each of these cases in turn.

Suppose that $e\in O_{1}$. A single iteration of the for loop of
$\stream$ ensures that $e$ is in $S_{i,1}\cup\dots\cup S_{i,m}$
with probability $p_{e}\geq\varepsilon$. Since we perform $r=\Theta(\ln(1/\varepsilon)/\varepsilon)$
independent iterations, we have $\pr{e\notin U}\leq(1-\varepsilon)^{r}\leq\exp(-\varepsilon r)\leq\varepsilon$.

Suppose that $e\in O_{2}$. We have
\begin{align*}
\pr{e\in O_{2}\setminus O'_{2}} & =\pr{e\in\stg\left(V_{1,1}\cup\{e\}\right)}=p_{e}\leq\varepsilon \enspace,
\end{align*}
where the first equality follows from the definition of $O'_{2}$,
the second equality follows from the definition of $p_{e}$ and the
fact that $V_{1,1}\sim\mathcal{V}(1/m)$, and the inequality follows
from the definition of $O_{2}$.
\end{proof}

Lemmas \ref{lem:O'2-gain} and \ref{lem:with-O'2} immediately imply
the following:

\newcommand{\LemSOneOOneVal}{We have $\Exp{f\left(S_{1,1}\cup(O_{1}\cap U)\right)}\geq(1-3\varepsilon)f(\opt)-\threshold b$.}
\begin{lemma}
\label{lem:S1-O1-val}
\LemSOneOOneVal
\end{lemma}
\begin{proof}
Recall that we use the notation $f(A\mid B)=f(A\cup B)-f(B)$. We
have
\begin{align*}
f\left(S_{1,1}\cup(O_{1}\cap U)\right)
&=f\left(S_{1,1}\cup(O_{1}\cap U)\cup O'_{2}\right)-f\left(O'_{2}\mid S_{1,1}\cup(O_{1}\cap U)\right)\\
&\geq f\left(S_{1,1}\cup(O_{1}\cap U)\cup O'_{2}\right)-f\left(O'_{2}\mid S_{1,1}\right) \enspace,
\end{align*}
where the inequality is by submodularity.

By taking expectation and using Lemmas \ref{lem:O'2-gain} and \ref{lem:with-O'2},
we obtain the desired result.
\end{proof}

Finally, Lemmas \ref{lem:T-val} and \ref{lem:S1-O1-val} give the
approximation guarantee:

\newcommand{\LemApprox}{We have $\Exp{\max\left\{ f(S_{1,1}),f(T)\right\} }\geq\left(\frac{\alpha}{1+\alpha}-3\varepsilon\right)f(\opt)$.}
\begin{lemma}
\label{lem:approx}
\LemApprox
\end{lemma}
\begin{proof}
By Lemmas \ref{lem:T-val} and \ref{lem:S1-O1-val}, we have
\begin{align*}
\Exp{\max\left\{ f(S_{1,1}),f(T)\right\} } & \geq\frac{\alpha}{1+\alpha\left(1-\frac{b}{k}\right)}\Exp{f\left(S_{1,1}\cup(O_{1}\cap U)\right)}\\
 & \geq\frac{\alpha}{1+\alpha\left(1-\frac{b}{k}\right)}\left((1-3\varepsilon)f(\opt)-\threshold b\right)\\
 & =\frac{\alpha}{1+\alpha\left(1-\frac{b}{k}\right)}\left((1-3\varepsilon)f(\opt)-\frac{\alpha}{1+\alpha}\frac{b}{k}f(\opt)\right)\\
 & \geq\left(\frac{\alpha}{1+\alpha}-3\varepsilon\right)f(\opt) \enspace.
\qedhere
\end{align*}
\end{proof}

\end{document}